\documentclass[12pt]{article}
\usepackage{amsmath,amsfonts,amsthm,amssymb,amscd,latexsym,multicol}
\usepackage[cp1251]{inputenc}
\usepackage[dvips]{graphicx}
\usepackage[dvips]{color}
\usepackage[all]{xypic}
\newtheorem{theorem}{Theorem}
\newtheorem{theorem*}{Theorem}
\newtheorem{lemma*}{Lemma}

\newtheorem{corollary*}{Corollary}

\newtheorem{remark*}{Remark}

\theoremstyle{definition}
\newtheorem{example}{Example}
\theoremstyle{remark}
\newtheorem{remark}{Remark}
\begin{document}
\begin{center}
{\Large\bf{Bayesian Post-Processor and other Enhancements of Subset
Simulation for Estimating Failure Probabilities in High Dimensions}}
\vspace{5mm}

Konstantin M. Zuev\footnote{Division of Engineering and Applied
Science, California Institute of Technology, Mail Code 104-44,
Pasadena, CA 91125, USA (Emails: zuev@caltech.edu,
jimbeck@caltech.edu)}, James L. Beck$^1$, Siu-Kui
Au\footnote{Department of Building and Construction, City University
of Hong Kong, 83 Tat Chee Avenue, Kowloon, Hong Kong (Email:
siukuiau@cityu.edu.hk)}, Lambros S. Katafygiotis\footnote{Department
of Civil and Environmental Engineering, The Hong Kong University of
Science and Technology, Hong Kong, China (Email: lambros@ust.hk)}
\end{center}

\vspace{3mm}
\begin{abstract}
Estimation of small failure probabilities is one of the most
important and challenging computational problems in reliability
engineering. The failure probability is usually given by an integral
over a high-dimensional uncertain parameter space that is difficult
to evaluate numerically. This paper focuses on enhancements to
Subset Simulation (SS), proposed by Au and Beck, which provides an
efficient algorithm based on MCMC (Markov chain Monte Carlo)
simulation for computing small failure probabilities for general
high-dimensional reliability problems. First, we analyze the
Modified Metropolis algorithm (MMA), an MCMC technique, which is
used in SS for sampling from high-dimensional conditional
distributions. The efficiency and accuracy of SS directly depends on
the ergodic properties of the Markov chains generated by MMA, which
control how fast the chain explores the parameter space. We present
some observations on the optimal scaling of MMA for efficient
exploration, and develop an optimal scaling strategy for this
algorithm when it is employed within SS. Next, we provide a
theoretical basis for the optimal value of the conditional failure
probability $p_0$, an important parameter one has to choose when
using SS. We demonstrate that choosing any $p_0\in[0.1, 0.3]$ will
give similar efficiency as the optimal value of $p_0$. Finally, a
Bayesian post-processor SS+ for the original SS method is developed
where the uncertain failure probability that one is estimating is
modeled as a stochastic variable whose possible values belong to the
unit interval. Simulated samples from SS are viewed as informative
data relevant to the system’s reliability. Instead of a single real
number as an estimate, SS+ produces the posterior PDF of the failure
probability, which takes into account both prior information and the
information in the sampled data. This PDF quantifies the uncertainty
in the value of the failure probability and it may be further used
in risk analyses to incorporate this uncertainty. To demonstrate
SS+, we consider its application to two different reliability
problems: a linear reliability problem and reliability analysis of
an elasto-plastic structure subjected to strong seismic ground
motion. The relationship between the original SS and SS+ is also
discussed.
\end{abstract}

\vspace{5mm}


\section{Introduction}
One of the most important and challenging problems in reliability
engineering is to estimate the failure probability $p_F$ for a
system, that is, the probability of unacceptable system performance.
This is usually expressed as an integral over a high-dimensional
uncertain parameter space:
\begin{equation}\label{pF}
p_F=\int I_F(\theta)\pi(\theta)d\theta=\mathbb{E}_\pi[I_F(\theta)],
\end{equation}
where $\theta\in \mathbb{R}^d$ represents the uncertain parameters
needed to specify completely the excitation and dynamic model of the
system; $\pi(\theta)$ is the joint probability density function
(PDF) for $\theta$; $F\subset \mathbb{R}^d$ is the failure domain in
the parameter space (i.e. the set of parameter values that lead to
performance of the system that is considered to be unacceptable);
and $I_F(\theta)$ stands for the indicator function, i.e.
$I_F(\theta)=1$ if $\theta\in F$ and $I_F(\theta)=0$ if
$\theta\notin F$. The dimension $d$ is typically large for dynamic
reliability problems (e.g. $d\sim 10^3$) because the stochastic
input time history is discretized in time. As a result, the usual
numerical quadrature methods for integrals are not computationally
feasible for evaluation (\ref{pF}).

Over the past decade, the engineering research community has
realized the importance of advanced stochastic simulation methods
for reliability analysis. As a result, many different efficient
algorithms have been developed recently, e.g. Subset Simulation
\cite{AuBeck}, Importance Sampling using Elementary Events
\cite{AuBeckISEE}, Line Sampling \cite{LS1}, Auxiliary domain method
\cite{ADM}, Spherical Subset Simulation \cite{KatafygiotisCheung2},
Horseracing Simulation \cite{HRS}, to name but a few.

This paper focuses on enhancements to Subset Simulation (SS),
proposed by Au and Beck in \cite{AuBeck}, which provides an
efficient algorithm for computing failure probabilities for general
high-dimensional reliability problems. It has been shown
theoretically \cite{AuBeck} and verified with different numerical
examples (e.g.
\cite{AuBeckSeismicRisk,AuChingBeck,SchuellerBeckAuKataGhanem}) that
SS gives much higher computational efficiency than standard Monte
Carlo Simulation when estimating small failure probabilities.
Recently, various modifications of SS have been proposed: SS with
splitting \cite{ChingAuBeck}, Hybrid SS \cite{ChingBeckAu}, and
Two-Stage SS \cite{KatafygiotisCheung}. It is important to
highlight, however, that none of these modifications offer a drastic
improvement over the original algorithm.

We start with the analysis of the Modified Metropolis algorithm
(MMA), a Markov chain Monte Carlo technique used in SS, which is
presented in Section \ref{Subset Simulation}. The efficiency and
accuracy of SS directly depends on the ergodic properties of the
Markov chains generated by MMA. In Section \ref{Optimal Scaling}, we
examine the optimal scaling of MMA to tune the parameters of the
algorithm to make the resulting Markov chain converge to
stationarity as fast as possible. We present a collection of
observations on the optimal scaling of MMA for different numerical
examples, and develop an optimal scaling strategy for MMA when it is
employed within SS for estimating small failure probabilities.

One of the most important components of SS which affects its
efficiency is the choice of the sequence of intermediate threshold
values or, equivalently, the intermediate failure probabilities (see
Section \ref{Subset Simulation}, where the original SS algorithm is
described). In Section \ref{optimal p0}, a method for optimally
choosing these probabilities is presented.

The usual interpretation of Monte Carlo methods is consistent with a
purely \textit{frequentist} approach, meaning that they can be
interpreted in terms of the frequentist definition of probability
which identifies it with the long-run relative frequency of
occurrence of an event. An alternative interpretation can be made
based on the \textit{Bayesian} approach which views probability as a
measure of the plausibility of a proposition conditional on
incomplete information that does not allow us to establish the truth
or falsehood of the proposition with certainty. Bayesian probability
theory was, in fact, primarily developed by the mathematician and
astronomer Laplace \cite{Laplace1,Laplace2} for statistical analysis
of astronomical observations. Moreover, Laplace developed the
well-known Bayes' theorem in full generality, while Bayes did it
only for a special case \cite{Bayes}. A complete development based
on Laplace's theory, with numerous examples of its applications, was
given by the mathematician and geophysicist Jeffreys \cite{Jeffreys}
in the early $20^{\textrm{th}}$ century. Despite its usefulness in
applications, the work of Laplace and Jeffreys on probability was
rejected in favor of the frequentist approach by most statisticians
until late last century. Because of the absence of a strong
rationale behind the theory at that time, it was perceived as
subjective and not rigorous by many statisticians. A rigorous logic
foundation for the Bayesian approach was given in the seminal work
of the physicist Cox \cite{Cox1,Cox2} and expounded by the physicist
Jaynes \cite{Jaynes2,Jaynes3}, enhancing Bayesian probability theory
as a convenient mathematical language for inference and uncertainty
quantification. Although the Bayesian approach usually leads to
high-dimensional integrals that often cannot be evaluated
analytically nor numerically by straightforward quadrature, the
development of Markov chain Monte Carlo algorithms and increasing
computing power have led over the past few decades to an explosive
growth of Bayesian papers in all research disciplines.

In Section \ref{Bayesian Subset Simulation} of this paper, a
Bayesian post-processor for the original Subset Simulation method is
developed, where the uncertain failure probability that one is
estimating is modeled as a stochastic variable whose possible values
belong to the unit interval. Although this failure probability is a
constant defined by the integral in (\ref{pF}), its exact value is
unknown because the integral cannot be evaluated; instead, we must
infer its value from available relevant information. Instead of a
single real number as an estimate, the post-processor, written as
SS+ (``SS-plus'') for short, produces the posterior PDF of the
failure probability, which takes into account both relevant prior
information and the information from the samples generated by SS.
This PDF expresses the relative plausibility of each possible value
of the failure probability based on this information. Since this PDF
quantifies the uncertainty in the value of $p_F$, it can be fully
used in risk analyses (e.g. for life-cycle cost analysis, decision
making under risk, etc.), or it can be used to give a point estimate
such as the most probable value based on the available information.


\section{Subset Simulation}\label{Subset Simulation}

\subsection{Basic idea of Subset Simulation}

The original and best known stochastic simulation algorithm for
estimating high-dimensional integrals is Monte Carlo Simulation
(MCS). In this method the failure probability $p_F$ is estimated by
approximating the mean of $I_F(\theta)$ in (\ref{pF}) by its sample
mean:
\begin{equation}\label{Monte Carlo}
p_F\approx \hat{p}_F^{MC}=\frac{1}{N}\sum_{i=1}^N
I_F(\theta^{(i)}),\hspace{3mm}
\end{equation}
where samples $\theta^{(1)},\ldots,\theta^{(N)}$ are independent and
identically distributed (i.i.d.) samples from $\pi(\cdot)$, denoted
$\theta^{(i)}\stackrel{i.i.d.}{\sim}\pi(\cdot)$. This estimate is
just the fraction of samples that produce system failure according
to a model of the system dynamics. Notice that each evaluation of
$I_F$ requires a deterministic system analysis to be performed to
check whether the sample implies failure. The main advantage of MCS
is that its efficiency does not depend on the dimension $d$ of the
parameter space. Indeed, straightforward calculation shows that the
coefficient of variation (c.o.v) of the Monte Carlo estimate
(\ref{Monte Carlo}), serving as a measure of accuracy in the usual
interpretation of MCS, is given by:
\begin{equation}\label{MC CV}
\delta(\hat{p}_F^{MC})=\sqrt{\frac{1-p_F}{Np_F}}.
\end{equation}
However, MCS has a serious drawback: it is inefficient in estimating
small failure probabilities. If $p_F$ is very small, $p_F\ll1$, then
it follows from (\ref{MC CV}) that the number of samples $N$ (or,
equivalently, number of system analyses) needed to achieve an
acceptable level of accuracy is very large, $N\propto 1/p_F \gg 1$.
This deficiency of MCS has motivated research to develop more
efficient stochastic simulation algorithms for estimating small
failure probabilities in high-dimensions.

The basic idea of Subset Simulation \cite{AuBeck} is the following:
represent a very small failure probability $p_F$ as a product of
larger probabilities so $p_F=\prod_{j=1}^m p_j$, where the factors
$p_j$ are estimated sequentially, $p_j\approx \hat{p}_{j}$, to
obtain an estimate $\hat{p}_F^{SS}$ for $p_F$ as
$\hat{p}_F^{SS}=\prod_{j=1}^m\hat{p}_{j}$. To reach this goal, let
us consider a decreasing sequence of nested subsets of the parameter
space, starting from the entire space and shrinking to the failure
domain $F$:
\begin{equation}
\mathbb{R}^d=F_0\supset F_1 \supset \ldots \supset F_{m-1} \supset
F_m=F.
\end{equation}
Subsets $F_1,\ldots, F_{m-1}$  are called \textit{intermediate
failure domains}. As a result, the failure probability $p_F=P(F)$
can be rewritten as a product of conditional probabilities:
\begin{equation}
p_F=\prod_{j=1}^{m}P(F_{j}|F_{j-1})=\prod_{j=1}^{m}p_j,
\end{equation}
where $p_j=P(F_j|F_{j-1})$ is the conditional probability at the
$(j-1)^{th}$ conditional level. Clearly, by choosing the
intermediate failure domains appropriately, all conditional
probabilities $p_j$ can be made large. Furthermore, they can be
estimated, in principle, by the fraction of independent conditional
samples that cause failure at the intermediate level:
\begin{equation}\label{pjMC}
p_j\approx \hat{p}_j^{MC}=\frac{1}{N}\sum_{i=1}^{N}
I_{F_j}(\theta^{(i)}_{j-1}),\hspace{5mm}
\theta^{(i)}_{j-1}\stackrel{i.i.d.}{\sim}\pi(\cdot|F_{j-1}).
\end{equation}
Hence, the original problem (estimation of the small failure
probability $p_F$) is replaced by a sequence of $m$ intermediate
problems (estimation of the larger failure probabilities $p_j$,
$j=1,\ldots,m$).

The first probability $p_1=P(F_1|F_0)=P(F_1)$ is straightforward to
estimate by MCS, since (\ref{pjMC}) requires sampling from
$\pi(\cdot)$ that is assumed to be readily sampled. However, if
$j\geq2$, to estimate $p_j$ using (\ref{pjMC}) one needs to generate
independent samples from conditional distribution
$\pi(\cdot|F_{j-1})$, which, in general, is not a trivial task. It
is not efficient to use MCS for this purpose, especially at higher
levels, but it can be done by a specifically tailored Markov chain
Monte Carlo technique at the expense of generating dependent
samples.

Markov chain Monte Carlo (MCMC) \cite{Liu,NealMCMC,RobertCasella} is
a class of algorithms for sampling from multi-dimensional target
probability distributions that cannot be directly sampled, at least
not efficiently. These methods are based on constructing a Markov
chain that has the distribution of interest as its stationary
distribution. By simulating samples from the Markov chain, they will
eventually be draws from the target probability distribution but
they will not be independent samples. In Subset Simulation, the
Modified Metropolis algorithm (MMA) \cite{AuBeck}, an MCMC technique
based on the original Metropolis algorithm
\cite{Metropolis,Hastings}, is used for sampling from the
conditional distributions $\pi(\cdot|F_{j-1})$.
\begin{remark} It was observed in \cite{AuBeck} that the original
Metropolis algorithm does not work in high-dimensional conditional
probability spaces, because it produces a Markov chain with very
highly correlated states. The geometrical reasons for this are
discussed in \cite{Geominsight}.
\end{remark}

\subsection{Modified Metropolis algorithm}

Suppose we want to generate a Markov chain with stationary
distribution
\begin{equation}
\pi(\theta|\mathbb{F})=\frac{\pi(\theta)I_{\mathbb{F}}(\theta)}{P(\mathbb{F})}=\frac{\prod_{k=1}^d\pi_k(\theta_k)I_{\mathbb{F}}(\theta)}{P(\mathbb{F})},
\end{equation}
where $\mathbb{F}\subset\mathbb{R}^d$ is a subset of the parameter
space. Without significant loss of generality, we assume here that
$\pi(\theta)=\prod_{k=1}^d\pi_k(\theta_k)$, i.e. components of
$\theta$ are independent (but are not so when conditioned on
$\mathbb{F}$). MMA differs from the original Metropolis algorithm
algorithm in the way the candidate state $\xi=(\xi_1,\ldots,\xi_d)$
is generated. Instead of using a $d$-variate proposal PDF on
$\mathbb{R}^d$ to directly obtain the candidate state, in MMA a
sequence of univariate proposal PDFs is used. Namely, each
coordinate $\xi_k$ of the candidate state is generated separately
using a univariate proposal distribution dependent on the coordinate
$\theta_k$ of the current state. Then a check is made whether the
$d$-variate candidate generated in such a way belongs to the subset
$\mathbb{F}$ in which case it is accepted as the next Markov chain
state; otherwise it is rejected and the current MCMC sample is
repeated. To summarize, the Modified Metropolis algorithm proceeds
as follows:

\vspace{5mm} \hrule height 0.6pt \rule{0pt}{4mm}\centerline
{\textbf{Modified Metropolis algorithm
\cite{AuBeck}}}\rule{0pt}{4mm}
 \hrule
 \vspace{1mm}
 \texttt{Input:}

\hspace{0.5cm}$\vartriangleright$ $\theta^{(1)}\in\mathbb{F}$,
initial state of a Markov chain;

\hspace{0.5cm}$\vartriangleright$ $N$, total number of states, i.e.
samples;

\hspace{0.5cm}$\vartriangleright$
$\pi_1(\cdot),\ldots,\pi_d(\cdot)$, marginal PDFs of
$\theta_1,\ldots,\theta_d$, respectively;

\hspace{0.5cm}$\vartriangleright$
$S_1(\cdot|\alpha),\ldots,S_d(\cdot|\alpha)$, univariate proposal
PDFs depending on a parameter $\alpha\in\mathbb{R}$

\hspace{0.8cm} and satisfying the symmetry property
$S_k(\beta|\alpha)=S_k(\alpha|\beta)$, $k=1,\ldots,d$.

 \texttt{Algorithm:}

\hspace{0.5cm} \textbf{for} $i=1,\ldots,N-1$ \textbf{do}

\hspace{1.2cm}$\%$ Generate a candidate state $\xi$:

\hspace{1.2cm} \textbf{for} $k=1,\ldots, d$ \textbf{do}

\hspace{1.9cm} Sample $\tilde{\xi}_k\sim S_k(\cdot|\theta^{(i)}_k)$

\hspace{1.9cm} Compute the acceptance ratio
\begin{equation}\label{ratio}
r=\frac{\pi_k(\tilde{\xi}_k)}{\pi_k(\theta^{(i)}_k)}
\end{equation}

\hspace{1.9cm} Accept or reject $\tilde{\xi}_k$ by setting
\begin{equation}
\xi_k=\left\{
        \begin{array}{ll}
          \tilde{\xi}_k, & \hbox{with probability } \min\{1,r\}; \\
          \theta^{(i)}_k & \hbox{with probability } 1-\min\{1,r\}.
        \end{array}
      \right.
\end{equation}

\hspace{1.2cm} \textbf{end for}

\hspace{1.2cm} Check whether $\xi\in\mathbb{F}$  by system analysis

\hspace{1.2cm} and accept or reject $\xi$ by setting
\begin{equation}\label{accept reject}
\theta^{(i+1)}=\left\{
                 \begin{array}{ll}
                   \xi, & \hbox{if } \xi\in\mathbb{F}; \\
                   \theta^{(i)}, & \hbox{if } \xi\notin \mathbb{F}.
                 \end{array}
               \right.
\end{equation}

\hspace{0.5cm} \textbf{end for}

\texttt{Output:}

\hspace{0.5cm}$\blacktriangleright$
$\theta^{(1)},\ldots,\theta^{(N)}$, $N$ states of a Markov chain
with stationary distribution $\pi(\cdot|\mathbb{F})$.
 \vspace{1mm}
 \hrule
\vspace{5mm}

Schematically, the Modified Metropolis algorithm is shown in
Figure~1. For completeness and the reader's convenience, the proof
that $\pi(\cdot|\mathbb{F})$ is the stationary distribution for the
Markov chain generated by MMA is given in the Appendix.

\begin{remark} The symmetry property
$S_k(\beta|\alpha)=S_k(\alpha|\beta)$ does not play a critical role.
If $S_k$ does not satisfy this property, by replacing the
``Metropolis'' ratio in (\ref{ratio}) by the ``Metropolis-Hastings''
ratio
\begin{equation}\label{ratioHastings}
r=\frac{\pi_k(\tilde{\xi}_k)S_k(\theta^{(i)}_k|\tilde{\xi}_k)}{\pi_k(\theta^{(i)}_k)S_k(\tilde{\xi}_k|\theta^{(i)}_k)},
\end{equation}
we obtain an MCMC algorithm referred to as the Modified
Metropolis--Hastings algorithm.
\end{remark}

Thus, if we run the Markov chain for sufficiently long (the burn-in
period), starting from essentially any ``seed''
$\theta^{(1)}\in\mathbb{F}$, then for large $N$ the distribution of
$\theta^{(N)}$ will be approximately $\pi(\cdot|\mathbb{F})$. Note,
however, that in any practical application it is very difficult to
check whether the Markov chain has reached its stationary
distribution. If the seed $\theta^{(1)}\sim\pi(\cdot|\mathbb{F})$,
then all states $\theta^{(i)}$ will be automatically distributed
according to the target distribution,
$\theta^{(i)}\sim\pi(\cdot|\mathbb{F})$, since it is the stationary
distribution for the Markov chain. This is called \textit{perfect
sampling} \cite{RobertCasella} and Subset Simulation has this
property because of the way the seeds are chosen \cite{AuBeck}.

Let us assume now that we are given a seed
$\theta^{(1)}_{j-1}\sim\pi(\cdot|F_{j-1})$, where $j=2,\ldots,m$.
Then, using MMA, we can generate a Markov chain with $N$ states
starting from this seed and construct an estimate for $p_j$ similar
to (\ref{pjMC}), where MCS samples are replaced by MCMC samples:
\begin{equation}\label{pjMCMC}
p_j\approx \hat{p}_j^{MCMC}=\frac{1}{N}\sum_{i=1}^{N}
I_{F_j}(\theta^{(i)}_{j-1}),\hspace{5mm}
\theta^{(i)}_{j-1}\sim\pi(\cdot|F_{j-1}).
\end{equation}
Note that all samples $\theta^{(1)}_{j-1},\ldots,\theta^{(N)}_{j-1}$
in (\ref{pjMCMC}) are identically distributed in the stationary
state of the Markov chain, but are not independent. Nevertheless,
these MCMC samples can be used for statistical averaging as if they
were i.i.d., although with some reduction in efficiency \cite{Doob}.
Namely, the more correlated
$\theta^{(1)}_{j-1},\ldots,\theta^{(N)}_{j-1}$ are, the less
efficient is the estimate (\ref{pjMCMC}). The correlation between
successive samples is due to proposal PDFs $S_k$, which govern the
generation of the next state of the Markov chain from the current
one in MMA. Hence, the choice of proposal PDFs $S_k$ controls the
efficiency of estimate (\ref{pjMCMC}), making this choice very
important. It was observed in \cite{AuBeck} that the efficiency of
MMA depends on the spread of proposal distributions, rather then on
their type. Both small and large spreads tend to increase the
dependence between successive samples, slowing the convergence of
the estimator. Large spreads may reduce the acceptance rate in
(\ref{accept reject}), increasing the number of repeated MCMC
samples. Small spreads, on the contrary, may lead to a reasonably
high acceptance rate, but still produce very correlated samples due
to their close proximity. To find the optimal spread of proposal
distributions for MMA is a non-trivial task which is discussed in
Section \ref{Optimal Scaling}.

\begin{remark} In \cite{MMHDR} another modification of the Metropolis algorithm,
called Modified Metropolis--Hastings algorithm with Delayed
Rejection (MMHDR), has been proposed. The key idea behind MMHDR is
to reduce the correlation between states of the Markov chain. A way
to achieve this goal is the following: whenever a candidate $\xi$ is
rejected in (\ref{accept reject}), instead of getting a repeated
sample $\theta^{(i+1)}=\theta^{(i)}$, as the case of MMA, a new
candidate $\tilde{\xi}$ is generated using a new set of proposal
PDFs $\tilde{S}_k$. Of course, the acceptance ratios (\ref{ratio})
for the second candidate have to be adjusted in order to keep the
target distribution stationary. In general, MMHDR generates less
correlated samples than MMA but it is computationally more
expensive.
\end{remark}

\subsection{Subset Simulation algorithm}

Subset Simulation uses the estimates (\ref{pjMC}) for $p_1$ and
(\ref{pjMCMC}) for $p_j$, $j\geq2$, to obtain the estimate for the
failure probability:
\begin{equation}\label{pFSS0}
p_F\approx
\hat{p}_F^{SS}=\hat{p}_1^{MC}\prod_{j=2}^m\hat{p}_j^{MCMC}
\end{equation}

The remaining ingredient of Subset Simulation that we have to
specify is the choice of intermediate failure domains $F_1,\ldots,
F_{m-1}$. Usually, performance of a dynamical system is described by
a certain positive-valued performance function
$g:\mathbb{R}^d\rightarrow \mathbb{R}_+$, for instance, $g(\theta)$
may represent some peak (maximum) response quantity when the system
model is subjected to the uncertain excitation $\theta$. Then the
failure region, i.e. unacceptable performance region, can be defined
as a set of excitations that lead to the exceedance of some
prescribed critical threshold $b$:
\begin{equation}
F=\{\theta\in\mathbb{R}^d : g(\theta)>b\}.
\end{equation}
The sequence of intermediate failure domains can then be defined as
\begin{equation}
F_j=\{\theta\in\mathbb{R}^d : g(\theta)>b_j\},
\end{equation}
where $0<b_1<\ldots<b_{m-1}<b_m=b$. Intermediate threshold values
$b_j$ define the values of the conditional probabilities
$p_j=P(F_j|F_{j-1})$ and, therefore, affect the efficiency of Subset
Simulation. In practical cases it is difficult to make a rational
choice of the $b_j$-values in advance, so the $b_j$ are chosen
adaptively (see (\ref{bj}) below) so that the estimated conditional
probabilities are equal to a fixed value $p_0\in(0,1)$. We will
refer to $p_0$ as the \textit{conditional failure probability}.

\vspace{5mm} \hrule height 0.6pt \rule{0pt}{4mm}\centerline
{\textbf{Subset Simulation algorithm \cite{AuBeck}}}\rule{0pt}{4mm}
 \hrule
 \vspace{1mm}
 \texttt{Input:}

\hspace{0.5cm}$\vartriangleright$ $p_0$, conditional failure
probability;

\hspace{0.5cm}$\vartriangleright$ $N$, number of samples per
conditional level.

\texttt{Algorithm:}

\hspace{0.5cm} Set $j=0$, number of conditional level

\hspace{0.5cm} Set $N_F(j)=0$, number of failure samples at level $j$

\hspace{0.5cm} Sample
$\theta_0^{(1)},\ldots,\theta_0^{(N)}\stackrel{i.i.d.}{\sim}\pi(\cdot)$

\hspace{0.5cm} \textbf{for} $i=1,\ldots,N$ \textbf{do}

\hspace{1.2cm} \textbf{if} $g^{(i)}=g(\theta_0^{(i)})>b$ \textbf{do}

\hspace{1.6cm} $N_F(j)\leftarrow N_F(j)+1$

\hspace{1.2cm} \textbf{end if}

\hspace{0.5cm} \textbf{end for}

\hspace{0.5cm} \textbf{while} $N_F(j)/N<p_0$ \textbf{do}

\hspace{1.7cm} $j\leftarrow j+1$

\hspace{1.7cm} Sort $\{g^{(i)}\}$: $g^{(i_1)}\leq
g^{(i_2)}\leq\ldots\leq g^{(i_N)}$

\hspace{1.7cm} Define
\begin{equation}\label{bj}b_j=\frac{g^{(i_{N-Np_0})}+g^{(i_{N-Np_0+1})}}{2}\end{equation}

\hspace{1.7cm} \textbf{for} $k=1,\ldots,Np_0$ \textbf{do}

\hspace{2.4cm} Starting from
$\theta_j^{(1),k}=\theta_{j-1}^{(i_{N-Np_0+k})}\sim\pi(\cdot|F_j)$,
generate $1/p_0$ states

\hspace{2.4cm} of a Markov chain
$\theta_j^{(1),k},\ldots,\theta_j^{(1/p_0),k}\sim\pi(\cdot|F_j)$,
using MMA.

\hspace{1.7cm} \textbf{end for}

\hspace{1.7cm} Renumber:
$\{\theta_j^{(i),k}\}_{k=1,i=1}^{Np_0,1/p_0} \mapsto
\theta_j^{(1)},\ldots,\theta_j^{(N)}\sim\pi(\cdot|F_j)$

\hspace{1.7cm} \textbf{for} $i=1,\ldots,N$ \textbf{do}

\hspace{2.4cm} \textbf{if} $g^{(i)}=g(\theta_j^{(i)})>b$ \textbf{do}

\hspace{2.8cm} $N_F(j)\leftarrow N_F(j)+1$

\hspace{2.4cm} \textbf{end if}

\hspace{1.7cm} \textbf{end for}

\hspace{0.5cm} \textbf{end while}

 \texttt{Output:}

\hspace{0.5cm}$\blacktriangleright$ $\hat{p}_F^{SS}$, estimate of
$p_F$:
\begin{equation}\label{pFSS}
\hat{p}_F^{SS}=p_0^j\frac{N_F(j)}{N}
\end{equation}
 \vspace{1mm}
 \hrule
\vspace{5mm}

Schematically, the Subset Simulation algorithm is shown in
Figure~\ref{SSalg}.

The adaptive choice of $b_j$-values in (\ref{bj}) guarantees, first,
that all seeds $\theta_{j}^{(1),k}$ are distributed according to
$\pi(\cdot|F_j)$ and, second, that the estimated conditional
probability $P(F_j|F_{j-1})$ is equal to $p_0$. Here, for
convenience, $p_0$ is assumed to be chosen such that $Np_0$ and
$1/p_0$ are positive integers, although this is not strictly
necessary. In \cite{AuBeck} it is suggested to use $p_0=0.1$. The
optimal choice of conditional failure probability is discussed in
the next section.

\begin{remark} Subset Simulation provides an efficient stochastic
simulation algorithm for computing failure probabilities for general
reliability problems without using any specific information about
the dynamic system other than an input-output model. This
independence of a system's inherent properties makes Subset
Simulation potentially useful for applications in different areas of
science and engineering where the notion of ``failure'' has its own
specific meaning, e.g. in Computational Finance to estimate the
probability that a stock price will drop below a given threshold
within a given period of time, in Computational Biology to estimate
the probability of gene mutation, etc.
\end{remark}

\section{Tuning of the Modified Metropolis algorithm}\label{Optimal Scaling}

The efficiency and accuracy of Subset Simulation directly depends on
the ergodic properties of the Markov chain generated by the Modified
Metropolis algorithm; in other words, on how fast the chain explores
the parameter space and converges to its stationary distribution.
The latter is determined by the choice of one-dimensional proposal
distributions $S_k$, which makes this choice very important. In
spite of this, the choice of proposal PDFs is still largely an art.
It was observed in \cite{AuBeck} that the efficiency of MMA is not
sensitive to the type of the proposal PDFs; however, it depends on
their spread (e.g. their variance).

Optimal scaling refers to the need to tune the parameters of the
algorithm to make the resulting Markov chain converge to
stationarity as fast as possible. The issue of optimal scaling was
recognized in the original paper by Metropolis et al
\cite{Metropolis}. Gelman, Roberts, and Gilks \cite{GRG} were the
first authors to obtain theoretical results on the optimal scaling
of the original Metropolis algorithm. They proved that for optimal
sampling from a high-dimensional Gaussian distribution, the
Metropolis algorithm should be tuned to accept approximately $23\%$
of the proposed moves only. Since then many papers have been
published on optimal scaling of the original Metropolis algorithm.
In this section, in the spirit of \cite{GRG}, we address the
following question which is of high practical importance: what is
the optimal variance of the univariate Gaussian proposal PDFs for
simulating a high-dimensional Gaussian distribution conditional on
some specific domain using MMA?

This section is organized as follows: in Subsection \ref{History} we
recall the original Metropolis algorithm and provide a brief
overview of existing results on its optimal scaling; 
in Subsection \ref{MMscaling} we present a collection of
observations on the optimal scaling of the Modified Metropolis
algorithm for different numerical examples, and discuss the optimal
scaling strategy for MMA when it is employed within Subset
Simulation for estimating small failure probabilities.

\subsection{Metropolis algorithm: a brief history of its optimal scaling}\label{History}

The Metropolis algorithm is the most popular class of MCMC
algorithms. Let $\pi$ be the target PDF on $\mathbb{R}^d$;
$\theta^{(n)}$ be the current state of the Markov chain; and
$S(\cdot|\theta^{(n)})$ be a symmetric (i.e.
$S(\alpha|\beta)=S(\beta|\alpha)$) $d$-variate proposal PDF
depending on $\theta^{(n)}$. Then the  Metropolis update
$\theta^{(n)}\rightarrow\theta^{(n+1)}$ of the Markov chain works as
follows: first, simulate a candidate state $\xi\sim
S(\cdot|\theta^{(n)})$; next, compute the acceptance probability
$a(\xi|\theta^{(n)})=\min\{1,\pi(\xi)/\pi(\theta^{(n)})\}$; and,
finally, accept $\xi$ as a next state of the Markov chain, i.e. set
$\theta^{(n+1)}=\xi$, with probability $a(\xi|\theta^{(n)})$ or
reject $\xi$ by setting $\theta^{(n+1)}=\theta^{(n)}$ with the
remaining probability $1-a(\xi|\theta^{(n)})$. It easy to prove that
such updating leaves $\pi$ invariant, i.e. if $\theta^{(n)}$ is
distributed according to $\pi$, then so is $\theta^{(n+1)}$. Hence
the chain will eventually converge to $\pi$ as its stationary
distribution. Practically it means that if we run the Markov chain
for a long time, starting from any $\theta^{(1)}\in\mathbb{R}^d$,
then for large $N$ the distribution of $\theta^{(N)}$ will be
approximately $\pi$.

The variance $\sigma^2$ of the proposal PDF $S$ turns out to have a
significant impact on the speed of convergence of the Markov chain
to its stationary distribution. Indeed, if $\sigma^2$ is small, then
the Markov chain explores its state space very slowly. On the other
hand, if $\sigma^2$ is large, the probability to accept a new
candidate state is very low and this results in a chain remaining
still for long periods of time. Since the Metropolis algorithm with
extremal values of the variance of the proposal PDF produce a chain
that explores its state space slowly, it is natural to expect the
existence of an optimal $\sigma^2$ for which the convergence speed
is maximized.

The importance of optimal scaling was already realized in the
landmark paper \cite{Metropolis} where the Metropolis algorithm was
first introduced. Metropolis et al developed an algorithm for
generating samples from the Boltzmann distribution for solving
numerical problems in statistical mechanics. In this work the
uniform proposal PDF was used,
$S(\xi|\theta^{(n)})=U_{[\theta^{(n)}-\alpha,\theta^{(n)}+\alpha]}(\xi)$,
and it was noted:
\begin{quote}
``It may be mentioned in this connection that the maximum
displacement $\alpha$ must be chosen with some care; if too large,
most moves will be forbidden, and if too small, the configuration
will not change enough. In either case it will then take longer to
come to equilibrium.''
\end{quote}

In \cite{Hastings} Hastings generalized the Metropolis algorithm.
Namely, he showed that the proposal distribution need not be
uniform, and it need not to be symmetric. In the latter case, the
acceptance probability must be slightly modified:
$a(\xi|\theta^{(n)})=\min\left\{1,\frac{\pi(\xi)S(\theta^{(n)}|\xi)}{\pi(\theta^{(n)})S(\xi|\theta^{(n)})}\right\}$.
The corresponding algorithm is called the Metropolis-Hastings
algorithm. Furthermore, Hastings emphasized that the original
sampling method has a general nature and can be applied in different
circumstances (not only in the framework of statistical mechanics)
and that Markov chain theory (which is absent in \cite{Metropolis})
is a natural language for the algorithm. Among other insights,
Hastings made the following useful yet difficult to implement
recommendation:
\begin{quote}
Choose a proposal distribution ``so that the sample point in one
step may move as large a distance as possible in the sample space,
consistent with a low rejection rate.''
\end{quote}

Historically, the tuning of the proposal's variance was usually
performed by trial-and-error, typically using rules of thumb of the
following form: select $\sigma^2$ such that the corresponding
acceptance rate, i.e. the average number of accepted candidate
states, is between $30\%$ and $70\%$. The rationale behind such rule
is that too low an acceptance rate means that the Markov chain has
many repeated samples, while too high an acceptance rate indicates
that the chain moves very slowly. Although qualitatively correct,
these rules suffered from the lack of theoretical justification for
the lower and upper bounds for the acceptance rate. The first
theoretical result on the optimal scaling of the Metropolis
algorithm was obtained by Gelman et al \cite{GRG}. It was proved
that in order to perform optimally in high dimensional spaces, the
algorithm should be tuned to accept as small as $23\%$ of the
proposed moves. This came as an unexpected and counter-intuitive
result. Indeed, this states that the Markov chain should stay still
about $77\%$ of the time in order to have the fastest convergence
speed. Let us formulate the main result more precisely.

Suppose that all components of $\theta\in\mathbb{R}^d$ are i.i.d,
i.e. the target distribution $\pi(\theta)$ has the product form,
$\pi(\theta)=\prod_{i=1}^d f(\theta_i)$, where the one-dimensional
density $f$ satisfies certain regularity conditions (namely, $f$ is
a $C^2$-function and $(\log f)'$ is Lipschitz continuous). Then the
\textit{optimal} ``random walk'' Gaussian proposal PDF
$S(\xi|\theta^{(n)})=\mathcal{N}(\xi|\theta^{(n)},\sigma^2\mathbb{I}_d)$
has the following properties:
\begin{enumerate}
  \item The optimal standard deviation is $\sigma\approx 2.4/\sqrt{dI}$, where
$I=\mathbb{E}_f[((\log
f)')^2]=\int_{-\infty}^\infty\frac{(f'(x))^2}{f(x)}dx$  measures the
``roughness'' of $f$. The smoother the density is, the smaller $I$
is, and, therefore, the larger $\sigma$ is. In particular, for a
one-dimensional case ($d=1$) and \textit{standard} Gaussian $f$
($I=1$): $\sigma\approx2.4$ (a surprisingly high value!);
  \item The acceptance rate of the corresponding Metropolis algorithm is approximately $44\%$ for $d=1$
and declines to $23\%$ when $d\rightarrow\infty$. Moreover, the
asymptotic optimality of accepting  $23\%$ of proposed moves is
approximately true for dimension as low as $d=6$.

\end{enumerate}

This result gives rise to the following useful heuristic strategy,
that is easy to implement: tune the proposal variance so that the
average acceptance rate is roughly $25\%$. In spite of the i.i.d.
assumption for the target components, this result is believed to be
robust and to hold under various perturbations of the target
distribution. Being aware of practical difficulties of choosing the
optimal $\sigma^2$, Gelman et al provided a very useful observation:

\begin{quote}
    ``Interestingly, if one cannot be optimal, it seems better to use too high a value of $\sigma$ than too
low.''
\end{quote}

\begin{remark} This observation is consistent with the numerical result
obtained recently in \cite{MMHDR}: an increased variance of the
second stage proposal PDFs improves the performance of the MMHDR
algorithm.
\end{remark}

Since the pioneering work \cite{GRG}, the problem of optimal scaling
has attracted the attention of many researchers and optimal scaling
results have been derived for other types of MCMC algorithms. For
instance, the Metropolis-adapted Langevin algorithm (MALA) was
studied in \cite{RobertsRosental} and it was proved that the
asymptotically optimal acceptance rate for MALA is approximately
$57\%$. For a more detailed overview of existing results on the
optimal scaling of the Metropolis algorithm see \cite{Bedard} and
references cited therein.


\subsection{Optimal scaling of the Modified Metropolis algorithm}\label{MMscaling}

In this subsection we address two questions: what is the optimal
variance $\sigma^2$ of the univariate Gaussian proposal PDFs
$S_k(\cdot|\mu)=\mathcal{N}(\cdot|\mu,\sigma^2)$, $k=1,\ldots,d$ for
simulating a high-dimensional conditional Gaussian distribution
$\pi(\cdot|{F})=\mathcal{N}(\cdot|0,\mathbb{I}_d)I_{{F}}(\cdot)/P({F})$
using the Modified Metropolis algorithm and what is the optimal
scaling strategy for Modified Metropolis when it is employed within
Subset Simulation for estimating small failure probabilities?

Let us first define what we mean by ``optimal'' variance. Let
$\theta^{(i),k}_{j-1}$ be the the $i^{\textrm{th}}$ sample in the
$k^{\textrm{th}}$ Markov chain at simulation level $j-1$. The
conditional probability $p_j=P(F_j|F_{j-1})$ is then estimated as
follows:
\begin{equation}\label{pjMCMC2}
p_j\approx \hat{p}_j=\frac{1}{N}\sum_{k=1}^{N_c}\sum_{i=1}^{N_s}
I_{F_j}(\theta^{{(i)},k}_{j-1}),\hspace{5mm}
\theta^{{(i)},k}_{j-1}\sim\pi(\cdot|F_{j-1}),
\end{equation}
where $N_c$ is the number of Markov chains and $N_s$ is the total
number of samples simulated from each of these chains, $N_s=N/N_c$,
so that the total number of Markov chain samples is $N$. An
expression for the coefficient of variation (c.o.v.) of $\hat{p}_j$,
derived in \cite{AuBeck}, is given by:
\begin{equation}\label{covpj}
    \delta_j=\sqrt{\frac{1-p_j}{Np_j}(1+\gamma_j)},
\end{equation}
where
\begin{equation}\label{gamma_j}
    \gamma_j=2\sum_{i=1}^{N_s-1}\left(1-\frac{i}{N_s}\right)\frac{R^{(i)}_j}{R^{(0)}_j},
\end{equation}
and
\begin{equation}\label{R}
    R^{(i)}_j=\mathbb{E}[I_{F_j}(\theta_{j-1}^{(1),k})I_{F_j}(\theta_{j-1}^{(1+i),k})]-p_j^2
\end{equation}
is the autocovariance of the stationary stochastic process
$X(i)=I_{F_j}(\theta_{j-1}^{(i),k})$ at lag $i$. The term
$\sqrt{(1-p_j)/Np_j}$ in (\ref{covpj}) is the c.o.v. of the MCS
estimator with $N$ independent samples. The c.o.v. of $\hat{p}_j$
can thus be considered as the one in MCS with an effective number of
independent samples $N/(1+\gamma_j)$. The efficiency of the
estimator using dependent MCMC samples ($\gamma_j>0$) is therefore
reduced compared to the case when the samples are independent
($\gamma_j=0$). Hence, $\gamma_j$ given by (\ref{gamma_j}) can be
considered as a measure of correlation between the states of a
Markov chain and smaller values of $\gamma_j$ imply higher
efficiency.

\begin{remark} Formula (\ref{covpj}) was derived
assuming that the Markov chain generated according to MMA is ergodic
and that the samples generated by different chains are uncorrelated
through the indicator function, i.e.
$\mathbb{E}[I_{F_j}(\theta)I_{F_j}(\theta')]-p_j^2=0$ if $\theta$
and $\theta'$ are from different chains. The latter, however, may
not be always true, since the seeds for each chain may be dependent.
Nevertheless, the expression in (\ref{covpj}) provides a useful
theoretical description of the c.o.v. of $\hat{p}_j$.
\end{remark}

\begin{remark}The autocovariance sequence $R^{(i)}_j$, $i=0,\ldots,N_s-1$,
needed for calculation of $\gamma_j$, can be estimated using the
Markov chain samples at the $(j-1)^{\textrm{th}}$ level by:

\begin{equation}\label{Rapprox}
R^{(i)}_j\approx
\frac{1}{N-iN_c}\sum_{k=1}^{N_c}\sum_{i'=1}^{N_s-i}I_{F_j}(\theta_{j-1}^{(i'),k})I_{F_j}(\theta_{j-1}^{(i'+i),k})
-\hat{p}_j^2.
\end{equation}
\end{remark}

Note that in general, $\gamma_j$ depends on the number of samples
$N_s$ in the Markov chain, the conditional probability $p_j$, the
intermediate failure domains $F_{j-1}$ and $F_j$, and the standard
deviation $\sigma_j$ of the  proposal PDFs
$S_k(\cdot|\mu)=\mathcal{N}(\cdot|\mu,\sigma^2_j)$. According to the
``basic'' description of the Subset Simulation algorithm given in
Section \ref{Subset Simulation}, $p_j=p_0$ for all $j$ and
$N_s=1/p_0$. The latter, as it has been already mentioned, is not
strictly necessary, yet convenient. In this subsection, the value
$p_0$ is chosen to be $0.1$, as in the original paper \cite{AuBeck}.
In this setting,  $\gamma_j$ depends only on the standard deviation
$\sigma_j$ and the geometry of $F_{j-1}$ and $F_j$. For a given
reliability problem (i.e. for a given performance function $g$ that
defines domains $F_j$ for all $j$), $\sigma_j^{\mathrm{opt}}$ is
said to be the \textit{optimal spread} of the proposal PDFs at level
$j$, if it minimizes the value of $\gamma_j$:
\begin{equation}\label{optSigma}
    \sigma_j^{\mathrm{opt}}=\arg\min_{\sigma_j>0}\gamma_j(\sigma_j)
\end{equation}
We will refer to $\gamma_j=\gamma_j(\sigma_j)$ as
\textit{$\gamma$-efficiency} of the Modified Metropolis algorithm
with proposal PDFs $\mathcal{N}(\cdot|\mu,\sigma^2_j)$ at level $j$.

Consider two examples of the sequence of intermediate failure
domains.
\begin{example}[Exterior of a ball] Let $\theta=re\in\mathbb{R}^d$,
where $e$ is a unit vector and $r=\|\theta\|$. For many reasonable
performance functions $g$, if $r$ is large enough, then $\theta\in
F=\{\theta\in\mathbb{R}^d : g(\theta)>b\}$, i.e. $\theta$ is a
failure point, regardless of $e$. Therefore, an exterior of a ball,
$\bar{B}_r=\{\theta\in\mathbb{R}^d : \|\theta\|\geq r\}$, can serve
as an idealized model of many failure domains. Define the
intermediate failure domains as follows:
\begin{equation}\label{ball}
F_j=\bar{B}_{r_j},
\end{equation}
where the radii $r_j$ are chosen such that $P(F_j|F_{j-1})=p_0$, i.e
$r_j^2=F^{-1}_{\chi^2_d}(1-p_0^{-j})$, where $F_{\chi^2_d}$ denotes
the cumulative distribution function (CDF) of the chi-square
distribution with $d$ degrees of freedom. The dimension $d$  is
chosen to be $10^3$.
\end{example}
\begin{example}[Linear case] Consider a linear reliability problem
with performance function $g(\theta)=a^T\theta+b$, where
$a\in\mathbb{R}^d$ and $b\in\mathbb{R}$ are fixed coefficients. The
corresponding intermediate failure domains $F_j$ are half-spaces
defined as follows:
\begin{equation}\label{linear}
F_j=\{\theta\in\mathbb{R}^d : \langle\theta,e_a\rangle\geq
\beta_j\},
\end{equation}
where $e_a=\frac{a}{\|a\|}$ is the unit normal to the hyperplane
specified by $g$,  and the values of $\beta_j$ are chosen such that
$P(F_j|F_{j-1})=p_0$, i.e $\beta_j=\Phi^{-1}(1-p_0^{-j})$, where
$\Phi$ denotes the CDF of the standard normal distribution. The
dimension $d$ is chosen to be $10^3$.
\end{example}

For both examples, $\gamma_j$ as a function of $\sigma_j$ is plotted
in Fig.~\ref{sigma_vs_gamma} and the approximate values of the
optimal spread $\sigma_j^{\mathrm{opt}}$ are given in Table~1 for
simulation levels $j=1,\ldots,6$. As expected, the optimal spread
$\sigma_j^{\mathrm{opt}}$ decreases when $j$ increases, and based on
the numerical values in Table~1, $\sigma_j^{\mathrm{opt}}$ seems to
converge to approximately $0.3$ and $0.4$ in Example $1$ and $2$,
respectively. The following
 properties of the function
$\gamma_j=\gamma_j(\sigma_j)$ are worth mentioning:
\begin{enumerate}
  \item[(i)] $\gamma_j$ increases very rapidly, when $\sigma_j$ goes to zero;
  \item[(ii)] $\gamma_j$ has a deep trough around the optimal value
$\sigma_j^{\mathrm{opt}}$, when $j$ is large (e.g., $j\geq4$).
\end{enumerate}
Interestingly, these observations are consistent with the statement
given in \cite{GRG} and cited above: if one cannot be optimal (due
to (ii), it is indeed difficult to achieve optimality), it is better
to use too high a value of $\sigma_j$ than too low.

The question of interest now is what gain in efficiency can we
achieve for a proper scaling of the Modified Metropolis algorithm
when calculating small failure probabilities? We consider the
following values of failure probability: $p_F=10^{-k}$,
$k=2,\ldots,6$. The c.o.v. of the failure probability estimates
obtained by Subset Simulation are given in Fig.~\ref{covBall} and
Fig.~\ref{covLinear} for Examples $1$ and $2$, respectively. The
dashed (solid) curves correspond to the case when $N=300$ ($N=1000$)
samples are used per each intermediate failure region. For
estimation of each value of the failure probability, two different
MMAs are used within SS: the optimal algorithm with
$\sigma_j=\sigma_j^{\mathrm{opt}}$ (marked with stars); and the
reference algorithm with $\sigma_j=1$ (marked with squares). The
corresponding c.o.v's are denoted by $\delta_{\mathrm{opt}}$ and
$\delta_1$, respectively. From Fig.~\ref{covBall} and
Fig.~\ref{covLinear} it follows that the smaller $p_F$, the more
important to scale MMA optimally. When $p_F=10^{-6}$, the optimal
c.o.v $\delta_{\mathrm{opt}}$ is approximately 
$80\%$  of the reference c.o.v. $\delta_1$ for both examples, when
$N=1000$. 

Despite its obvious usefulness, the optimal scaling of the Modified
Metropolis algorithm is difficult to achieve in practice. First, as
shown in Table~1, the values of the optimal spread
$\sigma_j^{\mathrm{opt}}$ are different for different reliability
problems. Second, even for a given reliability problem, to find
$\sigma_j^{\mathrm{opt}}$ is computationally expensive because of
property (ii) of $\gamma_j$; and our simulation results show that
the qualitative properties (i) and (ii) generally hold for different
reliability problems, not only for Examples $1$ and $2$. Therefore,
we look for heuristic to choose $\sigma_j$ that is easy to implement
and yet gives near optimal behavior.

It has been recognized for a long time that, when using an MCMC
algorithm, it is useful to monitor its acceptance rate. Both
$\gamma$-efficiency $\gamma_j$ and the acceptance rate $\rho_j$ at
level $j$ depend on $\sigma_j$. For Examples $1$ and $2$, the
approximate values of the acceptance rate that corresponds to the
reference value $\sigma_j=1$ and the optimal spread
$\sigma_j=\sigma_j^{\mathrm{opt}}$ are given in Table 2; $\gamma_j$
as a function of $\rho_j$ is plotted in Fig.~\ref{accrate_vs_gamma},
for simulation levels $j=1,\ldots,6$. A key observation is that,
contrary to (ii), $\gamma_j$ is very flat around the optimal
acceptance rate $\rho_j^{\mathrm{opt}}$, which is defined as the
acceptance rate that corresponds to the optimal spread, i.e.
$\rho_j^{\mathrm{opt}}=\rho_j(\sigma_j^{\mathrm{opt}})$.
Furthermore, according to our simulation results this behavior is
typical, and not specific just for the considered examples. This
observation gives rise to the following heuristic scaling strategy:
\begin{quote}
\textit{At simulation level $j\geq1$ select $\sigma_j$ such that the
corresponding acceptance rate $\rho_j$ is between $30\%$ and
$50\%$}.
\end{quote}

This strategy is  easy to implement in the context of Subset
Simulation. At each simulation level $j$, $N_c$ Markov chains are
generated. Suppose, we do not know the optimal spread
$\sigma_j^{\mathrm{opt}}$ for our problem. We start with a reference
value, say $\sigma_j^{1:n}=1$, for the first $n$ chains. Based only
on these $n$ chains, we calculate the corresponding acceptance rate
$\rho_j^{1:n}$. If $\rho_j^{1:n}$ is too low (i.e. it is smaller
than $30\%$) we decrease the spread and use
$\sigma_j^{n+1:2n}<\sigma_j^{1:n}$ for the next $n$ chains. If
$\rho_j^{1:n}$ is too large (i.e. it is larger than $50\%$) we
increase the spread and use $\sigma_j^{n+1:2n}>\sigma_j^{1:n}$ for
the next $n$ chains. We proceed like this until all $N_c$ Markov
chains have been generated. Note that according to this procedure,
$\sigma_j$ is kept constant within a single chain and it is changed
only between chains. Hence the Markovian property is not destroyed.
The described strategy guaranties that the corresponding scaling on
the Modified Metropolis algorithm is nearly optimal.

\section{Optimal choice of conditional failure probability $p_0$}\label{optimal p0}
The parameter $p_0$ governs how many intermediate failure domains
$F_j$ are needed to reach the target failure domain $F$, which in
turn affects the efficiency of Subset Simulation. A very small value
of the conditional failure probability means that fewer intermediate
levels are needed to reach $F$ but it results in a very large number
of samples $N$ needed at each level for accurate estimation of the
small conditional probabilities $p_j=P(F_j|F_{j-1})$. In the extreme
case when $p_0\leq p_F$, Subset Simulation reduces to the standard
Monte Carlo simulation. On the other hand, increasing the value of
$p_0$ will mean fewer samples are needed for accurate estimation at
each level but it will increase the number of intermediate
conditional levels $m$. In this section we provide a theoretical
basis for the optimal value of the conditional failure probability.

We wish to choose the value $p_0$ such that the coefficient of
variation (c.o.v) of the failure probability estimator
$\hat{p}_F^{SS}$ is as small as possible, for the same total number
of samples. In \cite{AuBeck}, an analysis of the statistical
properties of the Subset Simulation estimator is given. If the
conditional failure domains are chosen so that the corresponding
estimates of the conditional probabilities are all equal to $p_0$
and the same number of samples $N$ is used in the simulation at each
conditional level, then the c.o.v. of the estimator $\hat{p}_F^{SS}$
for a failure probability $p_F=p_0^m$ (requiring $m$ conditional
levels) is approximated by
\begin{equation}\label{CV}
    \delta^2\approx\frac{m(1-p_0)}{Np_0}(1+\bar{\gamma}),
\end{equation}
where $\bar{\gamma}$ is the average correlation factor over all
levels (assumed to be insensitive to $p_0$) that reflects the
correlation among the MCMC samples in each level and depends on the
choice of the spread of the proposal PDFs. Since the total number of
samples $N_T=mN$ and the number of conditional levels $m=\log
p_F/\log p_0$, (\ref{CV}) can be rewritten as follows:
\begin{equation}\label{CV2}
    \delta^2\approx\frac{1-p_0}{p_0(\log p_0)^2}\times \frac{(\log
p_F)^2}{N_T}(1+\bar{\gamma}).
\end{equation}
Note that for given target failure probability $p_F$ and the total
number of samples $N_T$, the second factor in (\ref{CV2}) does not
depend on $p_0$. Thus, minimizing the first factor to minimize the
c.o.v. $\delta$ yields the optimal value as $p_0^{opt}\approx0.2$.
Figure~\ref{deltavsp0} shows the variation of $\delta$ as a function
of $p_0$ according to (\ref{CV2}) for $p_F=10^{-3}$, $N_T=2000$, and
$\bar{\gamma}=0,2,4,6,8,10$. 
This figure indicates that $\delta$ is relatively insensitive to
$p_0$ around its optimal value. Note that the shape of the trend is
invariant with respect to $p_F$, $N_T$ and $\bar{\gamma}$ because
their effects are multiplicative. The figure shows that choosing
$0.1\leq p_0\leq0.3$ will practically lead to similar efficiency and
it is not necessary to fine tune the value of the conditional
failure probability $p_0$ as long as Subset Simulation is
implemented properly.

\section{Bayesian Post-Processor for Subset Simulation}\label{Bayesian Subset
Simulation}

In this section we develop a Bayesian post-processor SS+ for the
original Subset Simulation algorithm described in Section
\ref{Subset Simulation}, that provides more information about the
value of $p_F$ than a single point estimate.

Recall that in SS the failure probability $p_F$ is represented as a
product of conditional probabilities $p_j=P(F_j|F_{j-1})$, each of
which is estimated using (\ref{pjMC}) for $j=1$ and (\ref{pjMCMC})
for $j=2,\ldots,m$. Let $n_j$ denote the number of samples
$\theta^{(1)}_{j-1},\ldots,\theta^{(N)}_{j-1}$ that belong to subset
$F_j$. The estimate for probability $p_j$ is then:
\begin{equation}\label{pjn+}
\hat{p}_j=\frac{n_j}{N}
\end{equation}
and the estimate for the failure probability defined by (\ref{pF})
is:
\begin{equation}\label{pFSS1}
\hat{p}_F^{SS}=\prod_{j=1}^m\hat{p}_j=\prod_{j=1}^m\frac{n_j}{N}.
\end{equation}
In order to construct a Bayesian post-processor for SS, we have to
replace the frequentist estimates (\ref{pjn+}) in (\ref{pFSS1}) by
their Bayesian analogs. In other words, we have to treat all
$p_1,\ldots,p_m$ and $p_F$ as \textit{stochastic variables} and,
following the Bayesian approach, proceed as follows:
\begin{enumerate}
  \item Specify prior PDFs $p(p_j)$ for all $p_j=P(F_j|F_{j-1})$, $j=1,\ldots,m$;
  \item Update each prior PDF, using new data
$\mathcal{D}_{j-1}=\{\theta_{j-1}^{(1)},\ldots,\theta_{j-1}^{(N)}\sim\pi(\cdot|F_{j-1})\}$,
i.e. find the posterior PDFs $p(p_j|\mathcal{D}_{j-1})$ via Bayes'
theorem;
  \item Obtain the posterior PDF $p(p_F|\cup_{j=0}^{m-1}\mathcal{D}_j)$
of $p_F=\prod_{j=1}^mp_j$ from $p(p_1|\mathcal{D}_{0}),\ldots,$
$p(p_m|\mathcal{D}_{m-1})$.
\end{enumerate}
\begin{remark} The term ``stochastic variable'' is used rather then
``random variable'' to emphasize that it is a variable whose value
is uncertain, not random, based on the limited information that we
have available, and for which a probability model is chosen to
describe the relative plausibility of each of its
possible values \cite{Beck,Jaynes3}. 
The failure probability $p_F$ is a constant given by (\ref{pF})
which lies in $[0,1]$ but its exact value is unknown because the
integral cannot be evaluated exactly, and so we quantify the
plausibility of its values based on the samples that probe the
performance function $g$.
\end{remark}

To choose the prior distribution for each $p_j$, we use the
\textit{Principle of Maximum Entropy} (PME), introduced by Jaynes
\cite{Jaynes1}. The PME postulates that, subject to specified
constraints, the prior PDF $p$ which should be taken to represent
the prior state of knowledge is the one that gives the largest
measure of uncertainty, i.e. maximizes Shannon's entropy which for a
continuous variable is given by $H(p)=-\int_{-\infty}^\infty
p(x)\log p(x)dx$. Since the set of all possible values for each
stochastic variable $p_j$ is the unit interval, we impose this as
the only constraint for $p(p_j)$, i.e. $\mbox{supp}\;p(p_j)=[0,1]$.
It is well known that the uniform distribution  is the maximum
entropy distribution among all continuous distributions on $[0,1]$,
so
\begin{equation}\label{prior}
p(p_j)=1, \hspace{5mm} 0\leq p_j\leq1.
\end{equation}
\begin{remark} We could choose a more informative prior PDF, perhaps based on previous experience with the failure probabilities for similar systems.
If the amount of data is large (i.e. $N$ is large), however, then
the effect of the prior PDF on the posterior PDF will be negligible
if the likelihood function has a unique global maximum. This
phenomenon is usually referred to in the literature as the
``stability'' or ``robustness'' of Bayesian estimation.
\end{remark}

Since initial samples $\theta_0^{(1)},\ldots,\theta_0^{(N)}$ are
i.i.d. according to $\pi$, the sequence of zeros and ones,
$I_{F_1}(\theta_0^{(1)}),\ldots,I_{F_1}(\theta_0^{(N)})$, can be
considered as Bernoulli trials and, therefore, the likelihood
function $p(\mathcal{D}_0|p_1)$ is a binomial distribution where
$\mathcal{D}_0$ consists of the number of $F_1$-failure samples
$n_1=\sum_{k=1}^NI_{F_1}(\theta_0^{(k)})$ and the total number of
samples is $N$. Hence, the posterior distribution of $p_1$ is the
beta distribution $\mathcal{B}e(n_1+1,N-n_1+1)$ (e.g. \cite{Gelman})
with parameters $(n_1+1)$ and $(N-n_1+1)$, i.e.
\begin{equation}\label{f1}
p(p_1|\mathcal{D}_0)=\frac{p_1^{n_1}(1-p_1)^{N-n_1}}{\mathrm{B}(n_1+1,N-n_1+1)},
\end{equation}
which is actually the original Bayes' result \cite{Bayes}. The beta
function $\mathrm{B}$ in (\ref{f1}) is a normalizing constant. If
$j\geq2$, all MCMC samples
$\theta_{j-1}^{(1)},\ldots,\theta_{j-1}^{(N)}$ are distributed
according to $\pi(\cdot|F_{j-1})$, however, they are not
independent. Nevertheless, analogously to the frequentist case,
where we used these samples for statistical averaging (\ref{pjMCMC})
as if they were i.i.d., we can use an expression similar to
(\ref{f1}) as a good approximation for the posterior PDF
$p(p_j|\mathcal{D}_{j-1})$ for $j\geq2$, so:
\begin{equation}\label{fjposterior}
p(p_j|\mathcal{D}_{j-1})\approx\frac{p_j^{n_j}(1-p_j)^{N-n_j}}{\mathrm{B}(n_j+1,N-n_j+1)},
\hspace{5mm}j\geq1,
\end{equation}
where $n_j=\sum_{k=1}^NI_{F_j}(\theta_{j-1}^{(k)})$ is the number of
$F_j$-failure samples. Note that in Subset Simulation the MCMC
samples $\theta_{j-1}^{(1)},\ldots,\theta_{j-1}^{(N)}$ consist of
the states of multiple Markov chains with different initial seeds
obtained from previous conditional levels. This makes the
approximation (\ref{fjposterior}) more accurate in comparison with
the case of a single chain.
\begin{remark}
It is important to highlight that using the r.h.s of
(\ref{fjposterior}) as the posterior PDF $p(p_j|\mathcal{D}_{j-1})$
is equivalent to considering samples
$\theta_{j-1}^{(1)},\ldots,\theta_{j-1}^{(N)}$ as independent. This
probability model ignores information that the samples are generated
by an MCMC algorithm, just as it ignores the fact that the random
number generator used to generate these samples is, in fact, a
completely deterministic procedure. One corollary of this neglected
information is that generally
$\delta_{p(p_F)}<\delta_{\hat{p}_F^{SS}}$, where $\delta_{p(p_F)}$
is the c.o.v. of the posterior PDF of $p_F$ and
$\delta_{\hat{p}_F^{SS}}$  is the c.o.v. of the original SS
estimator $\hat{p}_F^{SS}$ (see also the numerical examples in
Section \ref{examples}). Notice, however, that the two c.o.v.s are
fundamentally different: $\delta_{p(p_F)}$ is defined based on
samples generated from a single run of SS, while the frequentist
c.o.v. $\delta_{\hat{p}_F^{SS}}$ is defined based on repeated runs
of SS (an infinite number of them!).
\end{remark}

The last step is to find the PDF of the product of stochastic
variables $p_F=\prod_{j=1}^mp_j$, given the distributions of all
factors $p_j$ by  (\ref{fjposterior}).

\begin{remark} Products of random (or stochastic) variables play a central role
in many different fields such as physics (interactive particle
systems), number theory (asymptotic properties of arithmetical
functions), statistics (asymptotic distributions of order
statistics), etc. The theory of products of random variables is well
covered in \cite{Products}.
\end{remark}

In general, to find the distribution of a product of stochastic
variables is a non-trivial task. A well-known result is Rohatgi’s
formula \cite{Rohatgi}: if $X_1$ and $X_2$ are continuous stochastic
variables with joint PDF $f_{X_1,X_2}$, then the PDF of $Y=X_1X_2$
is
\begin{equation}\label{simple}
f_Y(y)=\int_{-\infty}^{+\infty}f_{X_1,X_2}\left(x,\frac{y}{x}\right)\frac{1}{|x|}dx.
\end{equation}
This result is straightforward to derive but it is difficult to
implement, especially when the number of stochastic variables is
more than two. In the special case of a product of independent beta
variables, Tang and Gupta \cite{Tang} derived an exact
representation for the PDF and provided a recursive formula for
computing the coefficients of this representation.

\begin{theorem}[Tang and Gupta, \cite{Tang}]
Let $X_1,\ldots,X_m$ be independent beta variables, $X_j\sim
\mathcal{B}e(a_j,b_j)$, and $Y=X_1X_2\ldots X_m$, then the
probability density function of $Y$ can be written as follows:
\begin{equation}\label{fY}
f_Y(y)=\left(\prod_{j=1}^m\frac{\Gamma(a_j+b_j)}{\Gamma(a_j)}\right)y^{a_m-1}(1-y)^{\sum_{j=1}^mb_j-1}\cdot\sum_{r=0}^\infty
\sigma_r^{(m)}(1-y)^r, \hspace{3mm} 0<y<1,
\end{equation}
where $\Gamma$ is the gamma function and coefficients
$\sigma_r^{(m)}$ are defined by the following recurrence relation:
\begin{equation}\label{sigma}
\sigma_r^{(k)}=\frac{\Gamma(\sum_{j=1}^{k-1}b_j+r)}{\Gamma(\sum_{j=1}^{k}b_j+r)}\sum_{s=0}^r\frac{(a_k+b_k-a_{k-1})_s}{s!}\sigma_{r-s}^{(k-1)},
\hspace{3mm} r=0,1,\ldots, \hspace{2mm} k=2,\ldots,m,
\end{equation}
with initial values
\begin{equation}\label{initialvalues}
\sigma_0^{(1)}=\frac{1}{\Gamma(b_1)}, \hspace{5mm} \sigma_r^{(1)}=0
\mbox{ for } r\geq1.
\end{equation}
Here, for any real number $\alpha\in\mathbb{R}$,
$(\alpha)_s=\alpha(\alpha+1)\ldots(\alpha+s-1)=\frac{\Gamma(\alpha+s)}{\Gamma(\alpha)}$.
\end{theorem}
We can obtain the posterior PDF of stochastic variable $p_F$ by
applying this theorem to $p_F=\prod_{j=1}^m p_j$, where $p_j\sim
\mathcal{B}e(n_j+1,N-n_j+1)$.

Let $p^{SS+}(p_F|\cup_{j=0}^{m-1}\mathcal{D}_j)$ denote the
right-hand side of (\ref{fY}) with $a_j=n_j+1$ and $b_j=N-n_j+1$,
and $\hat{p}_{MAP}^{SS+}$ be the maximum a posteriori (MAP)
estimate, i.e. the mode of $p^{SS+}$:
\begin{equation}\label{MAP}
\hat{p}_{MAP}^{SS+}=\arg\max\limits_{p_F\in[0,1]}p^{SS+}(p_F|\cup_{j=0}^{m-1}\mathcal{D}_j).
\end{equation}
Since the mode of a product of independent stochastic variables is
equal to the product of the modes, and the mode of the beta variable
$X\sim \mathcal{B}e(a,b)$ is $(a-1)/(a+b-2)$, we have:
\begin{equation}\label{MAP=frequentistestimate}
\hat{p}_{MAP}^{SS+}=\mbox{mode}
\left(p^{SS+}(p_F|\cup_{j=0}^{m-1}\mathcal{D}_j)\right)=\prod_{j=1}^m\mbox{mode}\left(p(p_j|\mathcal{D}_{j-1})\right)=\prod_{j=1}^m\frac{n_j}{N}=\hat{p}_F^{SS}.
\end{equation}
Thus, the original estimate $\hat{p}_F^{SS}$ of failure probability
obtained in the original Subset Simulation algorithm is just the MAP
estimate $\hat{p}_{MAP}^{SS+}$ corresponding to the PDF $p^{SS+}$.
This is a consequence of the choice of a uniform prior.

Although (\ref{fY}) provides an exact expression for the PDF of $Y$,
it contains an infinite sum that must be replaced by a truncated
finite sum in actual computations. This means that in applications
one has to use an approximation of the posterior PDF $p^{SS+}$ based
on (\ref{fY}). An alternative approach is to approximate the
distribution of the product $Y=\prod_{j=1}^m X_j$ by a single beta
variable $\tilde{Y}\sim \mathcal{B}e(a,b)$, where the parameters $a$
and $b$ are chosen so that $\mathbb{E}[\tilde{Y}]=\mathbb{E}[Y]$ and
$\mathbb{E}[\tilde{Y}^k]$ is as close to $\mathbb{E}[Y^k]$ as
possible for $2\leq k\leq K$, for some fixed $K$. This idea was
first proposed in \cite{Tukey}. In general, the product of beta
variables does not follow the beta distribution, nevertheless,  it
was shown in \cite{Fan} that the product can be accurately
approximated by a beta variable even in the case of $K=2$.

\begin{theorem}[Fan, \cite{Fan}] Let $X_1,\ldots,X_m$ be independent beta variables, $X_j\sim
\mathcal{B}e(a_j,b_j)$, and $Y=X_1X_2\ldots X_m$, then $Y$ is
approximately distributed as $\tilde{Y}\sim\mathcal{B}e(a,b)$, where
\begin{equation}\label{ab}
a=\mu_1\frac{\mu_1-\mu_2}{\mu_2-\mu_1^2}, \hspace{5mm}
b=(1-\mu_1)\frac{\mu_1-\mu_2}{\mu_2-\mu_1^2},
\end{equation}
and
\begin{equation}\label{mu}
\mu_1=\mathbb{E}[Y]=\prod_{j=1}^m\frac{a_j}{a_j+b_j}, \hspace{5mm}
\mu_2=\mathbb{E}[Y^2]=\prod_{j=1}^m\frac{a_j(a_j+1)}{(a_j+b_j)(a_j+b_j+1)}.
\end{equation}
\end{theorem}
It is easy to check that if $\tilde{Y}\sim \mathcal{B}e(a,b)$ with
$a$ and $b$ given by (\ref{ab}), then the first two moments of
stochastic variables $Y$ and $\tilde{Y}$ coincide, i.e.
$\mathbb{E}[\tilde{Y}]=\mathbb{E}[Y]$ and
$\mathbb{E}[\tilde{Y}^2]=\mathbb{E}[Y^2]$. The accuracy of the
approximation $Y\dot{\sim}\mathcal{B}e(a,b)$ is discussed in
\cite{Fan}.

Using Theorem 2, we can therefore approximate the posterior
distribution $p^{SS+}$ of stochastic variable $p_F$ by the beta
distribution as follows:
\begin{equation}\label{PDFapprox}
p^{SS+}(p_F|\cup_{j=0}^{m-1}\mathcal{D}_j)\approx
\tilde{p}^{SS+}(p_F|\cup_{j=0}^{m-1}\mathcal{D}_j)=\mathcal{B}e(p_F|a,b),
\hspace{3mm} \mbox{i.e. }p_F\dot{\sim}\mathcal{B}e(a,b),
\end{equation}
where
\begin{equation}
a=\frac{\prod_{j=1}^m\frac{n_j+1}{N+2}\left(1-\prod_{j=1}^m\frac{n_j+2}{N+3}\right)}{\prod_{j=1}^m\frac{n_j+2}{N+3}-\prod_{j=1}^m\frac{n_j+1}{N+2}},
\hspace{5mm}
b=\frac{\left(1-\prod_{j=1}^m\frac{n_j+1}{N+2}\right)\left(1-\prod_{j=1}^m\frac{n_j+2}{N+3}\right)}{\prod_{j=1}^m\frac{n_j+2}{N+3}-\prod_{j=1}^m\frac{n_j+1}{N+2}}.
\end{equation}
Since the first two moments of $p^{SS+}$ and $\tilde{p}^{SS+}$ are
equal (this also means the c.o.v. of $p^{SS+}$ and $\tilde{p}^{SS+}$
are equal), we have:
\begin{equation}\label{E1}
\begin{split}
  \mathbb{E}_{\tilde{p}^{SS+}}[p_F]=& \mathbb{E}_{p^{SS+}}[p_F]=\prod_{j=1}^m\mathbb{E}_{p(p_j|\mathcal{D}_{j-1})}[p_j]=\prod_{j=1}^m
\frac{n_j+1}{N+2}, \\
\mathbb{E}_{\tilde{p}^{SS+}}[p_F^2]=&
\mathbb{E}_{p^{SS+}}[p_F^2]=\prod_{j=1}^m\mathbb{E}_{p(p_j|\mathcal{D}_{j-1})}[p_j^2]=\prod_{j=1}^m
\frac{(n_j+1)(n_j+2)}{(N+2)(N+3)}.
\end{split}
\end{equation}
Notice that
\begin{equation}\label{E2}
\lim_{N\rightarrow\infty}\mathbb{E}_{\tilde{p}^{SS+}}[p_F]=\lim_{N\rightarrow\infty}\hat{p}_F^{SS},
\hspace{5mm}\mbox{and }\hspace{2mm}
\mathbb{E}_{\tilde{p}^{SS+}}[p_F]\approx\hat{p}_F^{SS}, \mbox{ when
} N \mbox{ is large}
\end{equation}
so, the mean value of the approximation $\tilde{p}^{SS+}$ to the
posterior PDF $p^{SS+}$ is accurately approximated by the original
estimate $\hat{p}_F^{SS}$ of the failure probability $p_F$.

Let us now summarize the Bayesian post-processor of Subset
Simulation. From the algorithmic point of view, SS+ differs from SS
only in the produced output. Instead of a single real number as an
estimate of $p_F$, SS+ produces the posterior PDF
$p^{SS+}(p_F|\cup_{j=0}^{m-1}\mathcal{D}_j)$ of the failure
probability, which takes into account both prior information and the
sampled data $\cup_{j=0}^{m-1}\mathcal{D}_j$ generated by SS, while
its approximation
$\tilde{p}^{SS+}(p_F|\cup_{j=0}^{m-1}\mathcal{D}_j)$ is more
convenient for further computations. The posterior PDF $p^{SS+}$ and
its approximation $\tilde{p}^{SS+}$ are given by (\ref{fY}) and
(\ref{PDFapprox}), respectively, where
\begin{equation}\label{nj}
n_j=\left\{
      \begin{array}{ll}
        p_0N, & \hbox{if }j<m; \\
        N_F, & \hbox{if }j=m,
      \end{array}
    \right.
\end{equation}
and $m$ is the total number of intermediate levels in the run of the
algorithm. The relationship between SS and SS+ is given by
(\ref{MAP=frequentistestimate}) and (\ref{E2}). Namely, the original
estimate $\hat{p}_F^{SS}$ of the failure probability based on the
samples produced by the Subset Simulation algorithm coincides with
the MAP estimate corresponding to $p^{SS+}$, and it also accurately
approximates the mean of $\tilde{p}^{SS+}$. Also, the c.o.v. of
$p^{SS+}$ and $\tilde{p}^{SS+}$ coincide and can be computed using
the first two moments in (\ref{E1}).

\begin{remark} Note that to incorporate the
uncertainty in the value of $p_F$, one can use the full PDF
$\tilde{p}^{SS+}$ for life-cycle cost analyses, decision making
under risk, and so on, rather than just using a point estimate of
$p_F$. For instance, a performance loss function $\mathcal{L}$ often
depends on the failure probability. In this case one can calculate
an expected loss given by the following integral:
\begin{equation}\label{loss}
\mathbb{E}[\mathcal{L}(p_F)]=\int_0^1
  \mathcal{L}(p_F)\tilde{p}^{SS+}(p_F)dp_F,
\end{equation}
which takes into account the uncertainty in the value of the failure
probability.
\end{remark}

\begin{remark}
We note that a Bayesian post-processor for Monte Carlo evaluation of
the integral (\ref{pF}), denoted as MC+, can be obtained as a
special case of SS+. The posterior PDF for the failure probability
$p_F$ based on $N$ i.i.d. samples
$\mathcal{D}=\{\theta^{(1)},\ldots,\theta^{(N)}\}$ from $\pi(\cdot)$
is given by
\begin{equation}\label{MC+}
    p^{MC+}(p_F|\mathcal{D})=\mathcal{B}e(p_F|n+1,N-n+1),
\end{equation}
where $n=\sum_{k=1}^NI_F({\theta^{(k)}})$ is the number of failure
samples.
\end{remark}

\section{Illustrative Examples}\label{examples}
To demonstrate the Bayesian post-processor of Subset Simulation, we
consider its application to two different reliability problems: a
linear reliability problem and reliability analysis of an
elasto-plastic structure subjected to strong seismic ground motion.

\subsection{Linear Reliability Problem}

As a first example, consider a linear failure domain. Let
 $d=10^3$ be the dimension of the linear
problem and suppose $p_F=10^{-3}$ is the exact failure probability.
The failure domain $F$ is defined as
\begin{equation}\label{F-domain-linear}
    F=\{\theta\in\mathbb{R}^d : \langle\theta,e\rangle\geq
\beta\},
\end{equation}
where $e$ is a unit vector and $\beta=\Phi^{-1}(1-p_F)\approx3.09$
is the reliability index. This example is one where FORM
\cite{Melchers,Der} gives the exact failure probability in terms of
$\beta$. Note that $\theta^*=\beta e$ is the design point of the
failure domain $F$ \cite{Geominsight,Melchers}. The failure
probability estimate $\hat{p}_F^{SS}$ obtained by SS and the
approximation of the posterior PDF $\tilde{p}^{SS+}$ obtained by SS+
are given in Fig.~\ref{linear_example} based on a number of samples
$N=10^3$ at each level ($m=3$ levels were needed). Observe that
$\tilde{p}^{SS+}$ is quite narrowly focused (with the mean
$\mu_{\tilde{p}^{SS+}}=1.064\times10^{-3}$ and the c.o.v.
$\delta_{\tilde{p}^{SS+}}=0.16$) around
$\hat{p}_F^{SS}=1.057\times10^{-3}$, which is very close to the
exact value. Note that the frequentist c.o.v. of the original SS
estimator $\hat{p}_F^{SS}$ is $\delta_{\hat{p}_F^{SS}}=0.28$  (based
on 50 independent runs of the algorithm).

\subsection{Elasto-Plastic Structure Subjected to Ground Motion}
This example of a non-linear system is taken from \cite{Au}.
Consider a structure that is modeled as a 2D six-story
moment-resisting steel frame with two-node beam elements connecting
the joints of the frame. The floors are assumed to be rigid in-plane
and the joints are assumed to be rigid-plastic. The yield strength
is assumed to be $317$ MPa for all members. Under service load
condition, the floors and the roof are subjected to a
uniformly-distributed static span load of $24.7$ kN/m and $13.2$
kN/m, respectively. For the horizonal motion of the structure,
masses are lumped at the floor levels, which include contributions
from live loads and the dead loads from the floors and the frame
members. The natural frequencies of the first two modes of vibration
are computed to be $0.61$ Hz and $1.71$ Hz. Rayleigh damping is
assumed so that the first two modes have $2\%$ of critical damping.
For a full description of the structure, see \cite{Au}.

The structure is subject to uncertain earthquake excitations modeled
as a nonstationary stochastic process. To simulate a time history of
the ground motion acceleration for given moment magnitude $M$ and
epicentral distance $r$, a discrete-time white noise sequence
$W_j=\sqrt{2\pi/\Delta t}Z_j$, $j=1,\ldots,N_t$ is first generated,
where $\Delta t=0.03$ s is the sampling time, $N_t=1001$ is the
number of time instants (which corresponds to a duration of $30$ s),
and $Z_1,\ldots,Z_{N_t}$ are i.i.d. standard Gaussian variables. The
white noise sequence is then modulated (multiplied) by an envelope
function $e(t; M, r)$ at the discrete time instants. The discrete
Fourier transform is then applied to the modulated white-noise
sequence. The resulting spectrum is multiplied with a radiation
spectrum $A(f; M, r)$ \cite{Au}, after which the discrete inverse
Fourier transform is applied to transform the sequence back to time
domain to yield a sample for the ground acceleration time history.
The synthetic ground motion $a(t; Z, M, r)$ generated from the model
is thus a function of the Gaussian vector $Z=(Z_1,\ldots,Z_{N_t})^T$
and stochastic excitation model parameters $M$ and $r$. Here, $M=7$
and $r=50$ km are used. For more details about the ground motion
sampling, refer to \cite{Au}.

In this example, the uncertainty  arises from seismic excitations
and the uncertain parameters $\theta=Z$, the i.i.d. Gaussian
sequence $Z_1,\ldots,Z_{N_t}$ that generates the synthetic ground
motion. The system response of interest, $g(\theta)$, is defined to
be the peak (absolute) interstory drift ratio
$\delta_{\max}=\max_{i=1,\ldots,6} \delta_i$, where $\delta_i$ is
the maximum absolute interstory drift ratio of the $i^{\mbox{\tiny
th}}$ story within the duration of study, $30$ s. The failure domain
$F\subset \mathbb{R}^{N_t}$ is defined as the exceedance of peak
interstory drift ratio in any one of the stories within the duration
of study. That is
\begin{equation}\label{Au_failure}
    F=\{\theta\in\mathbb{R}^{N_t} : \delta_{\max}(\theta)>b\},
\end{equation}
where $b$ is some prescribed critical threshold. In this example,
$b=0.5\%$ is considered, which, according to \cite{vision},
corresponds to ``operational'' damage level. For this damage level,
the structure may have a small amount of yielding.

The failure probability is estimated to be equal to $p_F=8.9\times
10^{-3}$ (based on $4\times 10^4$ Monte Carlo samples). In the
application of Subset Simulation, three different implementation
scenarios are considered: $N=500$, $N=1000$, and $N=2000$ samples
are simulated at each conditional level. The failure probability
estimates $\hat{p}_F^{SS}$ obtained by SS for these scenarios and
the approximation of the corresponding posterior PDFs
$\tilde{p}^{SS+}$ obtained by SS+ are given in
Fig.~\ref{nonlinear_example}. Observe that the more samples used
(i.e. the more information about the system that is extracted), the
more narrowly $\tilde{p}^{SS+}$ is focused around $\hat{p}_F^{SS}$,
as expected. The coefficients of variation are
$\delta_{\tilde{p}^{SS+}}=0.190$, $\delta_{\tilde{p}^{SS+}}=0.134$,
and $\delta_{\tilde{p}^{SS+}}=0.095$ for $N=500$, $N=1000$, and
$N=2000$, respectively. The corresponding frequentist coefficients
of variation of the original SS estimator $\hat{p}_F^{SS}$ are
$\delta_{\hat{p}_F^{SS}}=0.303$, $\delta_{\hat{p}_F^{SS}}=0.201$,
and $\delta_{\hat{p}_F^{SS}}=0.131$ (based on 50 independent runs of
the algorithm). In 
SS+, the coefficient of variation
$\delta_{\tilde{p}^{SS+}}$ can be considered as a measure of
uncertainty, based on the generated samples.

\section{Conclusions}\label{Conclusions}

This paper focuses on enhancements to the Subset Simulation (SS)
method, an efficient algorithm for computing failure probabilities
for general high-dimensional reliability problems  proposed by Au
and Beck \cite{AuBeck}.

First, we explore MMA (Modified Metropolis algorithm), an MCMC
technique employed within SS. This exploration leads to the
following nearly optimal scaling strategy for MMA: at simulation
level $j\geq1$, select $\sigma_j$ such that the corresponding
acceptance rate $\rho_j$ is between $30\%$ and $50\%$.

Next, we provide a theoretical basis for the optimal value of the
conditional failure probability $p_0$. We demonstrate that choosing
any $p_0\in[0.1, 0.3]$ will  lead to similar efficiency and it is
not necessary to fine tune the value of the conditional failure
probability as long as SS is implemented properly.

Finally, a Bayesian extension $SS+$ of the original SS method is
developed. In SS+, the uncertain failure probability $p_F$ that one
is estimating is modeled as a stochastic variable whose possible
values belong to the unit interval. Instead of a single real number
as an estimate as in SS, SS+ produces the posterior PDF
$p^{SS+}(p_F)$ of the failure probability, which takes into account
both prior information and the information in the samples generated
by SS. This PDF quantifies the uncertainty in the value of $p_F$
based on the samples and prior information and it may be used in
risk analyses to incorporate this uncertainty, or its approximation
$\tilde{p}^{SS+}(p_F)$, which is more convenient for further
computations. The original SS estimate corresponds to the most
probable value in the Bayesian approach.

\section*{Acknowledgements}\label{Acknowledgements}

This work was supported by the National Science Foundation, under
award number EAR-0941374 to the California Institute of Technology.
This support is gratefully acknowledged. Any opinions, findings, and
conclusions or recommendations expressed in this material are those
of the authors and do not necessarily reflect those of the National
Science Foundation.

\section*{Appendix}\label{Appendix}

In this Appendix we give a detailed proof that
$\pi(\cdot|\mathbb{F})$ is the stationary distribution for the
Markov chain generated by the Modified Metropolis algorithm
described in Section \ref{Subset Simulation}.

\begin{theorem}[Au and Beck, \cite{AuBeck}]
Let $\theta^{(1)},\theta^{(2)},\ldots$ be the Markov chain generated
by the Modified Metropolis algorithm, then $\pi(\cdot|\mathbb{F})$
is a stationary distribution, i.e. if $\theta^{(i)}$ is distributed
according to $\pi(\cdot|\mathbb{F})$, then so is $\theta^{(i+1)}$.
\end{theorem}

\begin{proof} Let $K$ denote the transition kernel of the Markov chain generated by
MMA. From the structure of the algorithm it follows that K has the
following form:
\begin{equation}\label{kernel}
    K(d\theta^{(i+1)}|\theta^{(i)})=k(d\theta^{(i+1)}|\theta^{(i)})
+ r(\theta^{(i)})\delta_{\theta^{(i)}}(d\theta^{(i+1)}),
\end{equation}
where $k$ describes the transitions from  $\theta^{(i)}$ to
$\theta^{(i+1)}\neq\theta^{(i)}$,
$r(\theta^{(i)})=1-\int_\mathbb{F}k(d\theta'|\theta^{(i)})$ is the
probability of remaining at $\theta^{(i)}$, and $\delta_\theta$
denotes point mass at $\theta$ (Dirac measure). Note that if
$\theta^{(i+1)}\neq\theta^{(i)}$, then
$k(d\theta^{(i+1)}|\theta^{(i)})$ can be expressed as a product of
the component transitional kernels:
\begin{equation}\label{k}
k(d\theta^{(i+1)}|\theta^{(i)})=\prod_{j=1}^dk_j(d\theta^{(i+1)}_j|\theta^{(i)}_j),
\end{equation}
where $k_j$ is the transitional kernel for the $j^{\mathrm{th}}$
component of $\theta^{(i)}$. By definition of the algorithm
\begin{equation}\label{kj}
k_j(d\theta^{(i+1)}_j|\theta^{(i)}_j)=S_j(\theta_j^{(i+1)}|\theta_j^{(i)})\min\left\{1,\frac{\pi_j(\theta_j^{(i+1)})}{\pi_j(\theta_j^{(i)})}\right\}d\theta^{(i+1)}_j+r_j(\theta_j^{(i)})\delta_{\theta_j^{(i)}}(d\theta_j^{(i+1)})
\end{equation}

A sufficient condition for $\pi(\cdot|\mathbb{F})$ to be a
stationary distribution for $K$ is to satisfy the so-called
reversibility condition (also sometimes called the detailed balance
equation):
\begin{equation}\label{DB}
    \pi(d\theta^{(i)}|\mathbb{F})K(d\theta^{(i+1)}|\theta^{(i)})=\pi(d\theta^{(i+1)}|\mathbb{F})K(d\theta^{(i)}|\theta^{(i+1)})
\end{equation}
Indeed, given (\ref{DB}) and
$\theta^{(i)}\sim\pi(\cdot|\mathbb{F})$, the distribution of
$\theta^{(i+1)}$ is
\begin{equation}\label{longformula}
\begin{split}
  &p(d\theta^{(i+1)})=\int_{\mathbb{F}}\pi(d\theta^{(i)}|\mathbb{F})K(d\theta^{(i+1)}|\theta^{(i)})\\
  &=\int_{\mathbb{F}}\pi(d\theta^{(i+1)}|\mathbb{F})K(d\theta^{(i)}|\theta^{(i+1)})\\
  &=\pi(d\theta^{(i+1)}|\mathbb{F})\int_{\mathbb{F}}K(d\theta^{(i)}|\theta^{(i+1)})\\
  &=\pi(d\theta^{(i+1)}|\mathbb{F}),
\end{split}
\end{equation}
since $\int_{\mathbb{F}}K(d\theta^{(i)}|\theta^{(i+1)})\equiv1$.

So, all we need to prove is that the transition kernel $K$ of the MM
algorithm satisfies the reversibility condition (\ref{DB}). Since
all the Markov chain samples lie in $\mathbb{F}$, it is sufficient
to consider the transition only between states in $\mathbb{F}$.
Hence, without loss of generality, we assume that both
$\theta^{(i)}$ and $\theta^{(i+1)}$ belong to $\mathbb{F}$. In
addition, we assume that $\theta^{(i)}\neq\theta^{(i+1)}$, since
otherwise (\ref{DB}) is trivial. It then follows from (\ref{kernel})
that in this case
$K(d\theta^{(i+1)}|\theta^{(i)})=k(d\theta^{(i+1)}|\theta^{(i)})$.
Taking into account (\ref{k}), we can rewrite the reversibility
condition (\ref{DB}) in terms of components:
\begin{equation}\label{reducedDP}
    \prod_{j=1}^d\pi_j(\theta_j^{(i)})k_j(d\theta^{(i+1)}_j|\theta^{(i)}_j)d\theta_j^{(i)}=
    \prod_{j=1}^d\pi_j(\theta_j^{(i+1)})k_j(d\theta^{(i)}_j|\theta^{(i+1)}_j)d\theta_j^{(i+1)}.
\end{equation}
Therefore, it is enough to show that for all $j=1,\ldots,d$
\begin{equation}\label{comp}
    \pi_j(\theta_j^{(i)})k_j(d\theta^{(i+1)}_j|\theta^{(i)}_j)d\theta_j^{(i)}=\pi_j(\theta_j^{(i+1)})k_j(d\theta^{(i)}_j|\theta^{(i+1)}_j)d\theta_j^{(i+1)}.
\end{equation}
If $\theta_j^{(i)}=\theta_j^{(i+1)}$, then (\ref{comp}) is trivial.
Assume that $\theta_j^{(i)}\neq\theta_j^{(i+1)}$. Then,  it follows
from (\ref{kj}) that
$k_j(d\theta^{(i+1)}_j|\theta^{(i)}_j)=S_j(\theta_j^{(i+1)}|\theta_j^{(i)})\min\left\{1,\frac{\pi_j(\theta_j^{(i+1)})}{\pi_j(\theta_j^{(i)})}\right\}d\theta^{(i+1)}_j$,
and (\ref{comp}) reduces to
\begin{equation}\label{final}
\pi_j(\theta_j^{(i)})S_j(\theta_j^{(i+1)}|\theta_j^{(i)})\min\left\{1,\frac{\pi_j(\theta_j^{(i+1)})}{\pi_j(\theta_j^{(i)})}\right\}=
\pi_j(\theta_j^{(i+1)})S_j(\theta_j^{(i)}|\theta_j^{(i+1)})\min\left\{1,\frac{\pi_j(\theta_j^{(i)})}{\pi_j(\theta_j^{(i+1)})}\right\}.
\end{equation}
Since $S_j$ is symmetric and $b\min\{1,a/b\}=a\min\{1,b/a\}$ for any
positive numbers $a$ and $b$, (\ref{final}) is satisfied. This
proves that $\pi(\cdot|\mathbb{F})$ is \textit{a} stationary
distribution of the MMA Markov chain.
\end{proof}

\begin{remark} A stationary distribution is unique and, therefore,
is the limiting distribution for a Markov chain, if the chain is
aperiodic and irreducible (see, for example, \cite{Tierney}). In the
case of MMA, aperiodicity is guaranteed by the fact that the
probability of having a repeated sample
$\theta^{(i+1)}=\theta^{(i)}$ is not zero. A Markov chain with
stationary distribution $p(\cdot)$ is irreducible if, for any
initial state, it has positive probability of entering any set to
which $p(\cdot)$ assigns positive probability. It is clear that MMA
with ``standard'' proposal distributions (e.g. Gaussian, uniform,
log-normal, etc) generates irreducible Markov chains. In this case,
$\pi(\cdot|\mathbb{F})$ is therefore the unique stationary
distribution of the MMA Markov chain.
\end{remark}

\newpage
\begin{table}[h]
\begin{center}
\begin{tabular}{lcccccc}
\hline
  Simulation Level $j$ & 1 & 2 & 3 & 4 & 5 & 6 \\
\hline\hline
  Example 1, $\sigma_j^{\mathrm{opt}}$& 0.9 & 0.7 & 0.4 & 0.3 & 0.3 & 0.3 \\
  Example 2, $\sigma_j^{\mathrm{opt}}$& 1.1 & 0.8 & 0.6 & 0.4 & 0.4 & 0.4 \\
  \hline
\end{tabular}
\end{center}
\label{Tab1}\caption{\footnotesize Approximate values of the optimal
spread for different simulation levels}
\end{table}

\begin{table}[h]
\begin{center}
\begin{tabular}{lcccccc}
\hline
  Simulation Level $j$ & 1 & 2 & 3 & 4 & 5 & 6 \\
\hline\hline
  Example 1, $\rho_j(1)$ &49\%  &30\%  &20\% &13\% &9\%  &6\%  \\
  Example 1, $\rho_j^{\mathrm{opt}}$&  51\%&39\%  &47\%  &50\%  &44\%  &40\%  \\
\hline
  Example 2, $\rho_j(1)$& 53\%&35\% &23\% &16\% &11\%  &8\%  \\
  Example 2, $\rho_j^{\mathrm{opt}}$& 52\% & 41\%& 40\% &49\%  &43\%  &37\%  \\
  \hline
\end{tabular}
\end{center}
\label{Tab22}\caption{\footnotesize Approximate values of the
acceptance rates $\rho_j(1)$ and $\rho_j^{\mathrm{opt}}$ that
correspond to the reference value $\sigma_j=1$ and the optimal
spread $\sigma_j=\sigma_j^{\mathrm{opt}}$, respectively}
\end{table}
\newpage

\begin{figure}\label{MMalg}\centering
\begin{picture}(150,150)
\put(0,10){\vector(1,0){150}} \put(10,0){\vector(0,1){150}}
\put(145,15){$\theta_k$} \put(15,145){$\theta_l$}
\qbezier(40,140)(30,30)(150,60) \put(50,120){$\mathbb{F}$}
\put(85,85){\circle*{5}}\put(68,72){$\theta^{(i)}$}
\multiput(83.5,13)(0,20){4}{$|$}
\multiput(10,82)(20,0){4}{\textbf{---}} \qbezier(40,10)(55,9)(65,19)
\qbezier(65,19)(85,39)(105,19) \qbezier(105,19)(115,9)(130,10)
\put(85,10){\circle*{3}}\put(68,-5){$\theta^{(i)}_k$}
\put(103,10){\circle*{3}}\put(98,-2){$\xi_k$}
 \qbezier(10,40)(9,55)(19,65)
\qbezier(19,65)(39,85)(19,105) \qbezier(19,105)(9,115)(10,130)
\put(10,85){\circle*{3}}\put(-5,85){$\theta^{(i)}_{l}$}
\put(10,105){\circle*{3}}\put(0,105){$\xi_{l}$}
\multiput(10,102)(20,0){5}{\textbf{---}}
\multiput(101.3,14)(0,20){5}{$|$} \put(103,105){\circle*{5}}
\put(107,107){$\theta^{(i+1)}$} 
\put(25,20){\footnotesize{$S_k(\cdot|\theta_k^{(i)})$}}
\put(15,50){\footnotesize{$S_l(\cdot|\theta_l^{(i)})$}}
\end{picture}
\caption{\footnotesize Modified Metropolis algorithm}\label{TTT}
\end{figure}
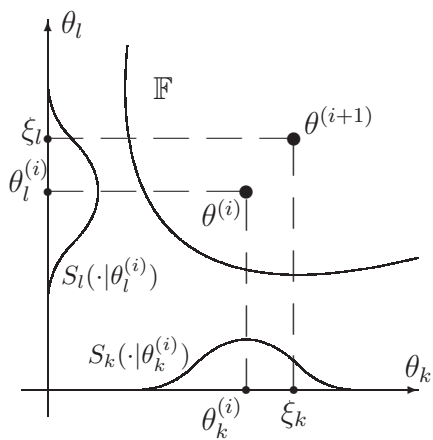
\newpage  \vspace{5cm}

\begin{figure}\label{SSalg}\centering
\begin{picture}(200,200)
\put(0,100){\vector(1,0){200}} \put(100,0){\vector(0,1){200}}

\qbezier(75,105)(65,125)(85,115)\qbezier(85,115)(115,105)(115,110)
\qbezier(115,110)(115,125)(120,105)\qbezier(120,105)(125,75)(95,80)
\qbezier(95,80)(80,85)(75,105)

\qbezier(50,110)(30,150)(70,130) \qbezier(70,130)(130,110)(130,150)
\qbezier(130,150)(135,170)(140,110)\qbezier(140,110)(140,50)(90,60)
\qbezier(90,60)(60,70)(50,110)

\qbezier(25,105)(-5,175)(55,145) \qbezier(55,145)(135,105)(125,175)
\qbezier(125,175)(125,200)(160,110)\qbezier(160,110)(180,45)(85,40)
\qbezier(85,40)(40,40)(25,105)

\put(118,115){\footnotesize{$F_{j-1}$}}\put(141,115){\footnotesize{$F_{j}$}}\put(160,115){\footnotesize{$F_{j+1}$}}
\put(75,75){\circle*{3}}\put(70,82){\circle*{3}}\put(125,75){\circle*{3}}\put(130,83){\circle*{3}}\put(127,90){\circle*{3}}\put(110,70){\circle*{3}}\put(120,68){\circle*{3}}
\put(90,120){\circle*{3}}\put(80,121){\circle*{3}}\put(70,125){\circle*{3}}
\put(110,115){\circle*{3}}\put(133,133){\circle*{3}}\put(60,120){\circle*{3}}\put(65,105){\circle*{3}}\put(85,70){\circle*{3}}

\put(75,132){\circle*{3}}\put(75,132){\circle{5}}\put(68,133){\circle{3}}\put(66,137){\circle{3}}\put(70,140){\circle{3}}\put(72,150){\circle{3}}\put(75,160){\circle{3}}\put(80,155){\circle{3}}\put(83,150){\circle{3}}\put(80,147){\circle{3}}\put(77,142){\circle{3}}\put(81,140){\circle{3}}
\put(40,110){\circle*{3}}\put(40,110){\circle{5}}\put(43,115){\circle{3}}\put(37,104){\circle{3}}\put(35,108){\circle{3}}\put(30,111){\circle{3}}\put(27,117){\circle{3}}\put(24,112){\circle{3}}\put(29,106){\circle{3}}\put(37,114){\circle{3}}\put(18,114){\circle{3}}\put(16,107){\circle{3}}
\put(65,65){\circle*{3}}\put(65,65){\circle{5}}\put(60,60){\circle{3}}\put(57,64){\circle{3}}\put(57,64){\circle{3}}\put(55,57){\circle{3}}\put(50,60){\circle{3}}\put(62,54){\circle{3}}\put(58,52){\circle{3}}\put(65,57){\circle{3}}\put(50,50){\circle{3}}\put(45,55){\circle{3}}\put(45,55){\circle{3}}\put(45,50){\circle{3}}
\put(140,75){\circle*{3}}\put(140,75){\circle{5}}\put(142,80){\circle{3}}\put(145,70){\circle{3}}\put(136,68){\circle{3}}\put(141,60){\circle{3}}\put(142,65){\circle{3}}\put(150,75){\circle{3}}\put(152,65){\circle{3}}\put(146,62){\circle{3}}\put(155,70){\circle{3}}\put(155,58){\circle{3}}
\put(139,145){\circle*{3}}\put(139,145){\circle{5}}\put(137,151){\circle{3}}\put(140,153){\circle{3}}\put(138,157){\circle{3}}\put(132,160){\circle{3}}\put(136,162){\circle{3}}\put(130,170){\circle{3}}\put(128,165){\circle{3}}\put(128,176){\circle{3}}\put(132,180){\circle{3}}
\end{picture}
\caption{\footnotesize Subset Simulation algorithm: disks $\bullet$
and circles $\circ$ represent samples from $\pi(\cdot|F_{j-1})$ and
$\pi(\cdot|F_{j})$, respectively; circled disks are the Markov chain
seeds for $\pi(\cdot|F_j)$}\label{SSalg}
\end{figure}
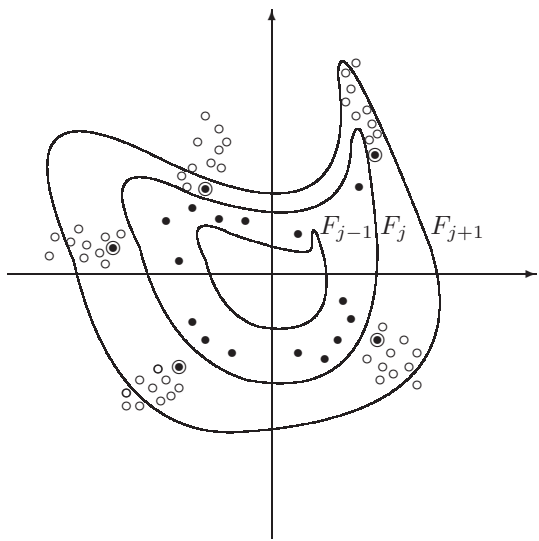
\newpage

\begin{figure}\centering
\includegraphics[angle=0,scale=1]{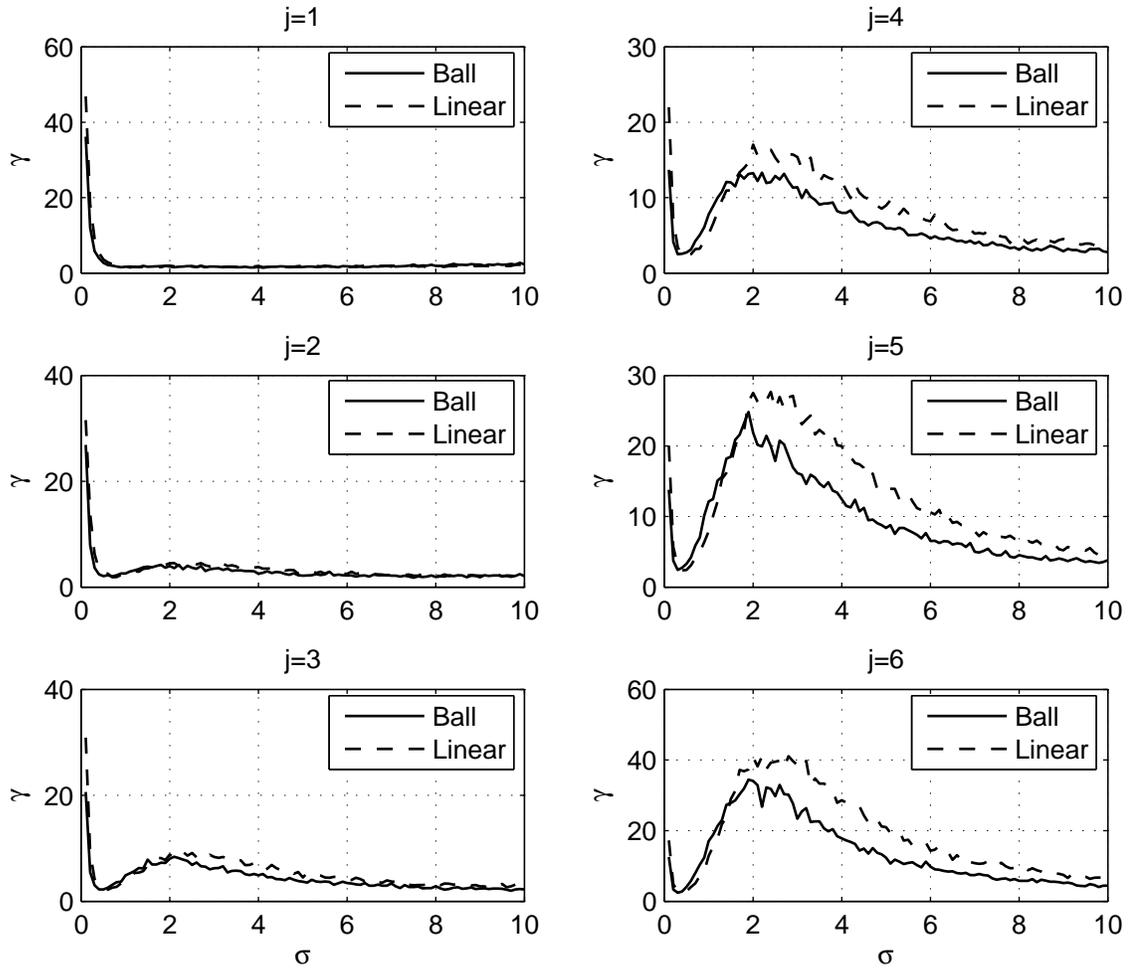}
\caption{\footnotesize The $\gamma$-efficiency of the Modified
Metropolis algorithm as a function of spread $\sigma$ for simulation
levels $j=1,\ldots,6$} \label{sigma_vs_gamma}
\end{figure}
\newpage

\begin{figure}\centering
\includegraphics[angle=0,scale=1]{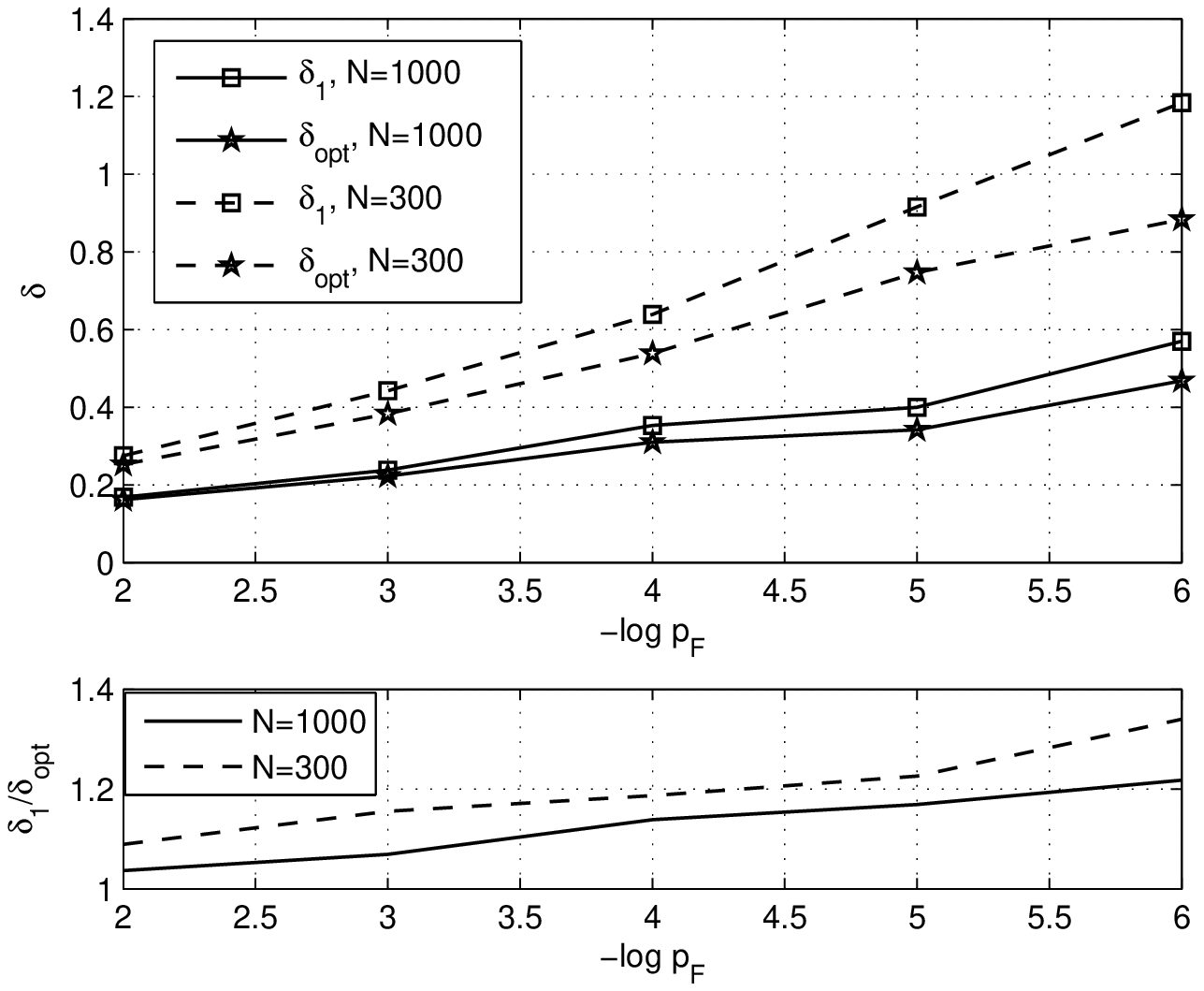}
\caption{\footnotesize The c.o.v. of $p_F$ estimates obtained by
Subset Simulation for Example $1$} \label{covBall}
\end{figure}
\newpage

\begin{figure}\centering
\includegraphics[angle=0,scale=1]{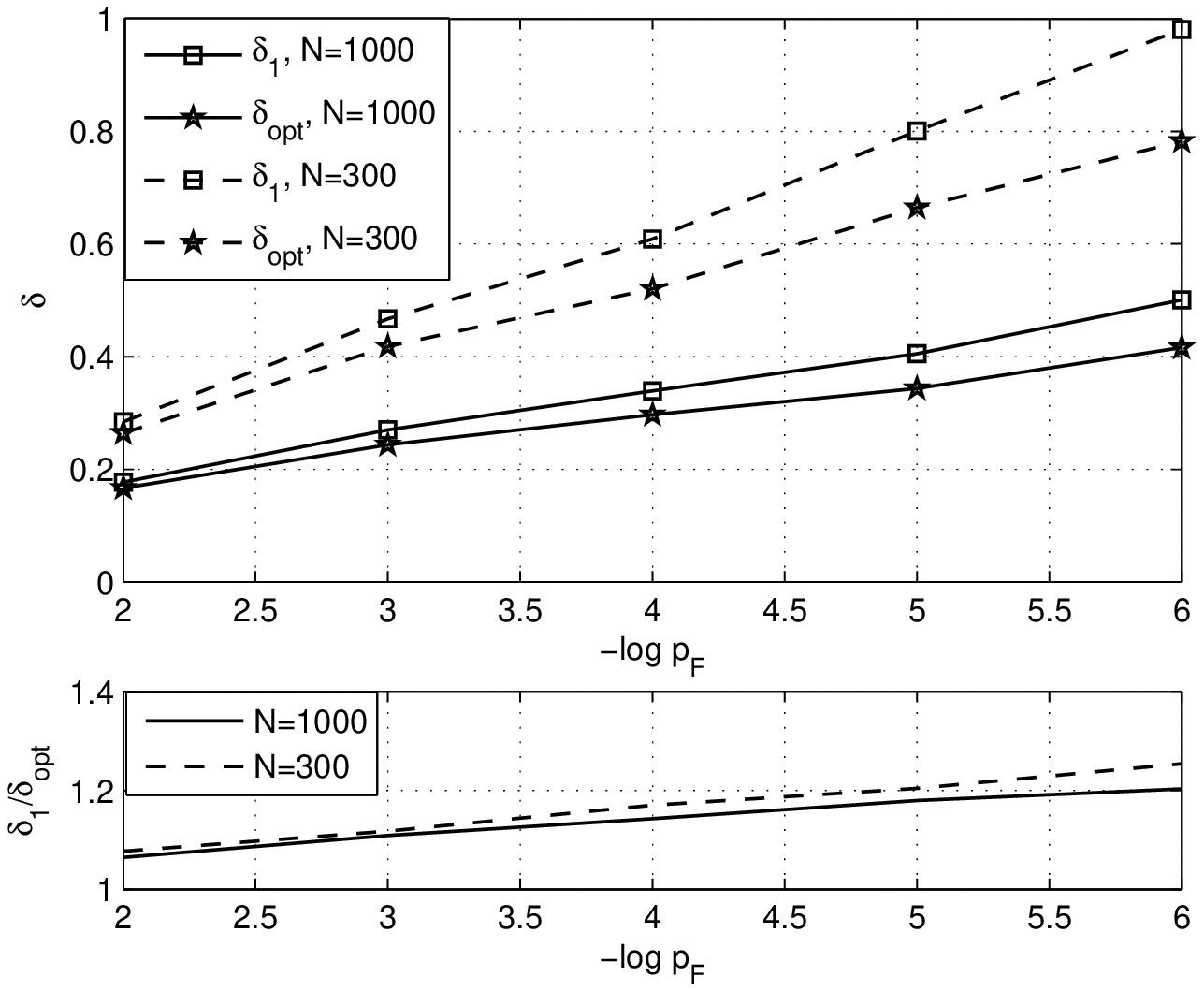}
\caption{\footnotesize The c.o.v. of $p_F$ estimates obtained by
Subset Simulation for Example $2$} \label{covLinear}
\end{figure}
\newpage

\begin{figure}\centering
\includegraphics[angle=0,scale=1]{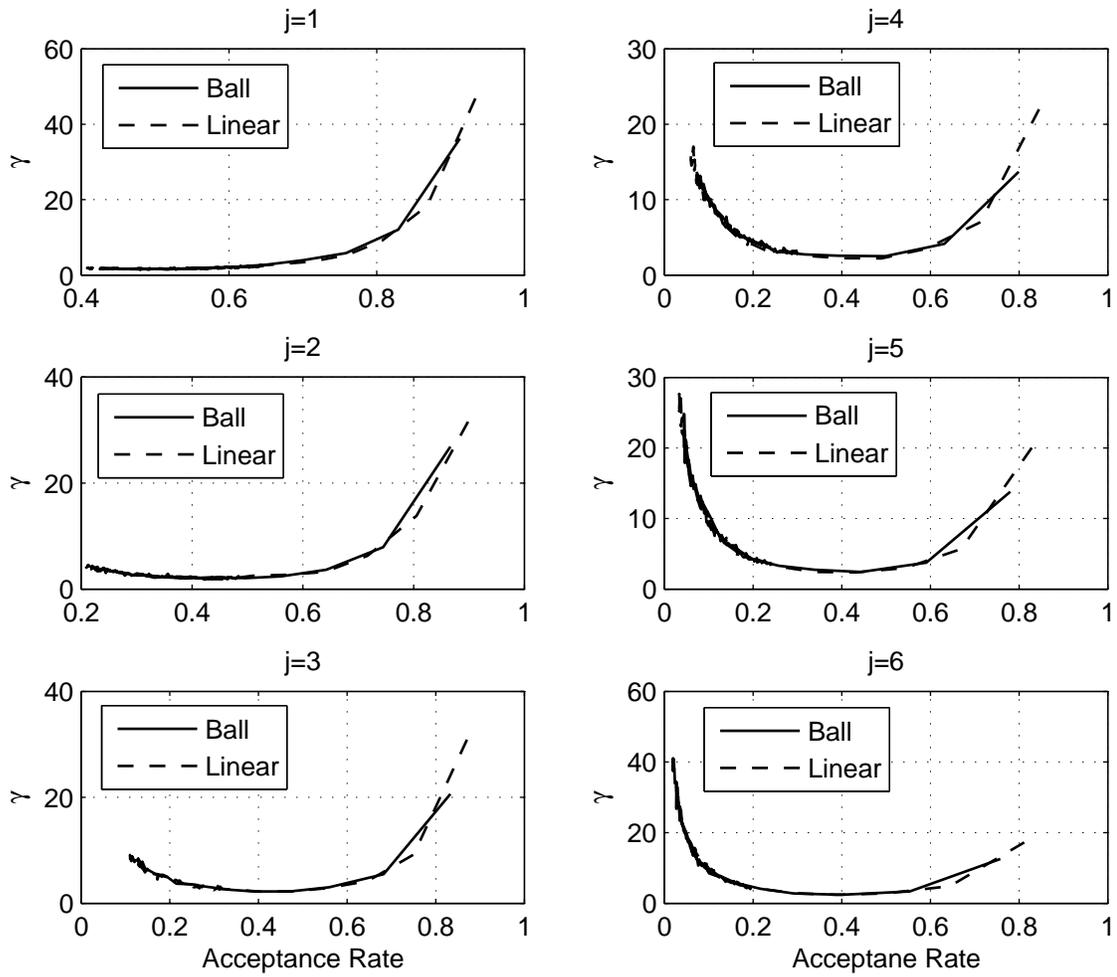}
\caption{\footnotesize The $\gamma$-efficiency of the Modified
Metropolis algorithm as a function of the acceptance rate for
simulation levels $j=1,\ldots,6$} \label{accrate_vs_gamma}
\end{figure}
\newpage

\begin{figure}\centering
\includegraphics[angle=0,scale=1]{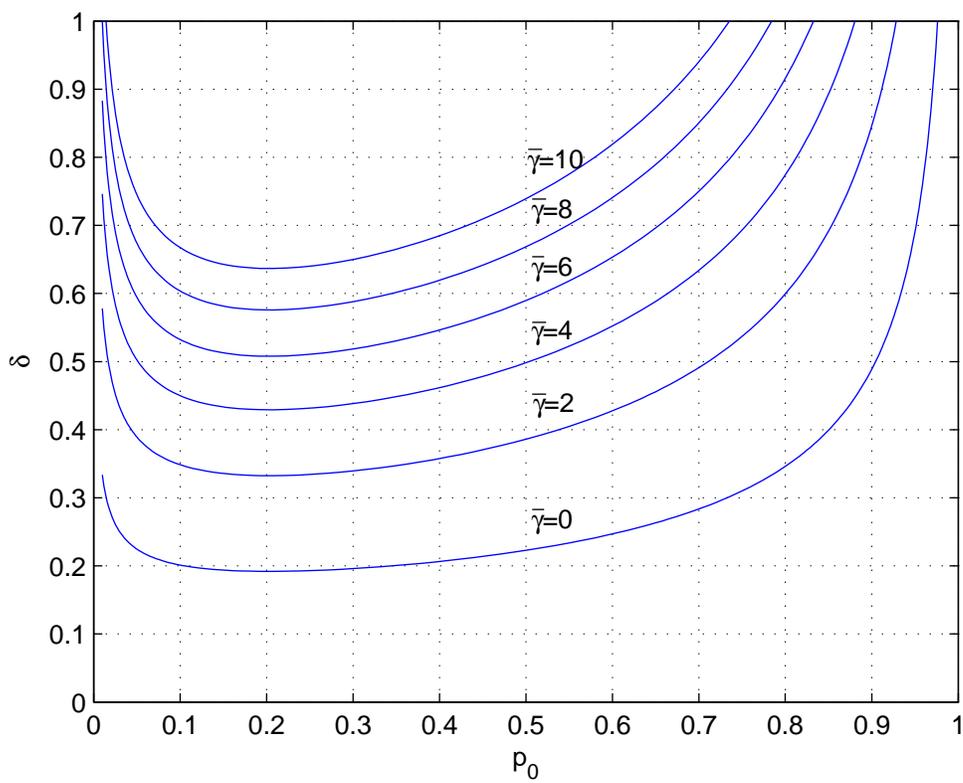}
\caption{\footnotesize Variation of $\delta$ as a function of $p_0$
according to (\ref{CV2}) for $p_F=10^{-3}$, $N_T=2000$, and
$\bar{\gamma}=0,2,4,6,8,10$} \label{deltavsp0}
\end{figure}
\newpage

\begin{figure}\centering
\includegraphics[angle=0,scale=1]{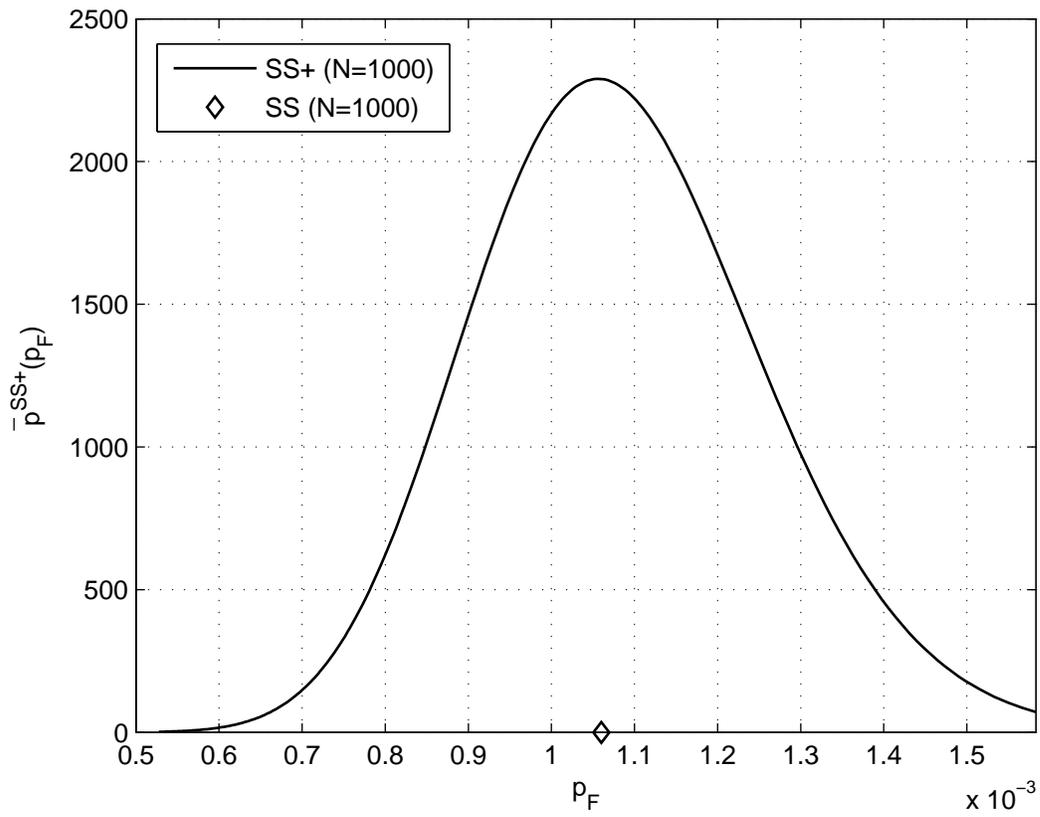}
\caption{\footnotesize The failure probability estimate
$\hat{p}_F^{SS}$ obtained by SS and the approximation of the
posterior PDF $\tilde{p}^{SS+}$ obtained by SS+. The posterior PDF
has mean $\mu_{\tilde{p}^{SS+}}=1.064\times10^{-3}$ and the c.o.v.
$\delta_{\tilde{p}^{SS+}}=0.16$} \label{linear_example}
\end{figure}

\begin{figure}\centering
\includegraphics[angle=0,scale=1]{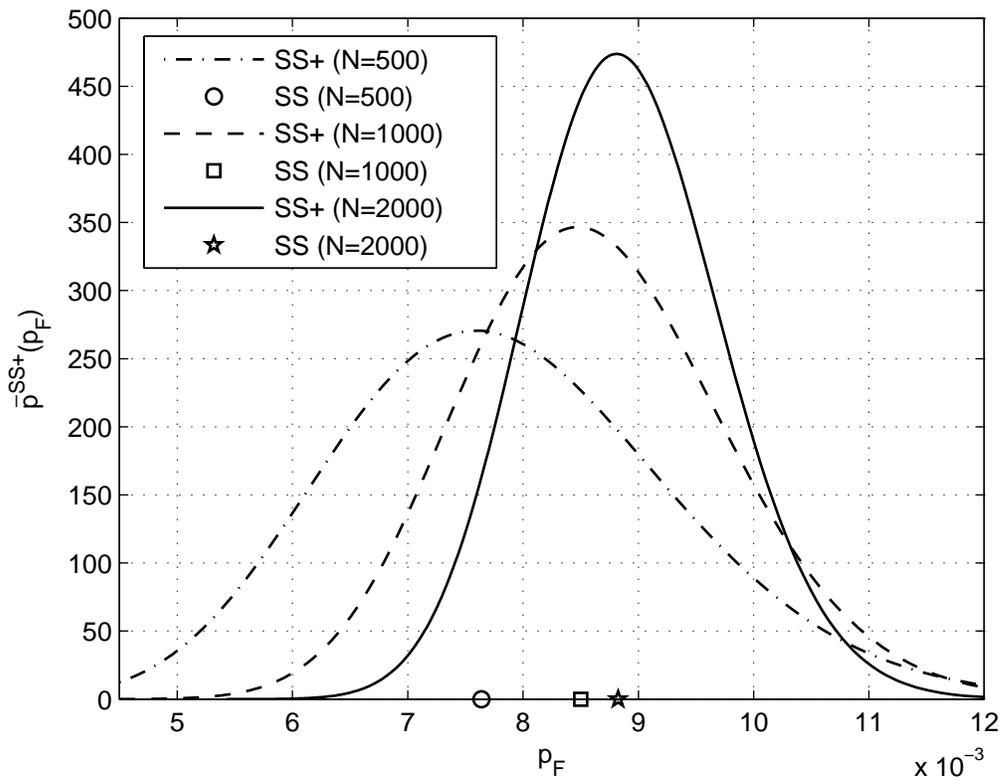}
\caption{\footnotesize The failure probability estimates
$\hat{p}_F^{SS}$ obtained by SS and the approximation of the
posterior PDF $\tilde{p}^{SS+}$ obtained by SS+ for three
computational scenarios: $N=500$, $N=1000$, and $N=2000$ samples at
each conditional level.} \label{nonlinear_example}
\end{figure}

\end{document}